\documentclass{siamltex}

\usepackage[breaklinks,colorlinks=true,linkcolor=blue,citecolor=red]{hyperref}
\usepackage[round]{natbib}
\usepackage{appendix}

\usepackage{amsfonts} 
\usepackage{amssymb,amsmath} 
\usepackage{graphicx,color}
  
\DeclareMathAlphabet{\itbf}{OML}{cmm}{b}{it}

\def\EE{\mathbb{E}}

\def\RR{\mathbb{R}}

\def\eps{\varepsilon}

\newcommand{\ea}{\end{eqnarray}}  
\newcommand{\ba}{\begin{eqnarray}}  
\newcommand{\ee}{\end{equation}}  
\newcommand{\be}{\begin{equation}}  
\newcommand{\ean}{\end{eqnarray*}}  
\newcommand{\ban}{\begin{eqnarray*}}

\begin{document}  
\bibliographystyle{plainnat}
\title{Option pricing under fast-varying long-memory stochastic volatility}

\author{Josselin Garnier\footnotemark[1]
 and Knut S\O lna\footnotemark[2]}   

\maketitle

\renewcommand{\thefootnote}{\fnsymbol{footnote}}

\footnotetext[1]{Centre de Math\'ematiques Appliqu\'ees,
Ecole Polytechnique, 91128 Palaiseau Cedex, France
{\tt josselin.garnier@polytechnique.edu}}

\footnotetext[2]{Department of Mathematics, 
University of California, Irvine CA 92697
{\tt ksolna@math.uci.edu}}

\renewcommand{\thefootnote}{\arabic{footnote}}

\maketitle

\begin{abstract}
Recent empirical studies suggest that the volatility of an underlying
price process may have correlations that decay slowly
under certain market conditions.  
In this paper, the volatility is modeled as a stationary process with long-range correlation properties
in order to capture such a situation, and we consider European option pricing.
This means that the volatility process
is neither a Markov process nor a martingale.  However,  
  by exploiting the fact that the price process is still a semimartingale 
and accordingly  using the martingale method, we can obtain an analytical expression
for the option price in the regime where the volatility process is fast mean-reverting.
The volatility process is modeled as a smooth and bounded function of a  fractional Ornstein--Uhlenbeck 
process.
We give the expression for the implied volatility, which has a fractional term structure.
\end{abstract}



\section{Introduction}

\subsection*{Stochastic Volatility and the Implied Surface}
Under many market scenarios, the assumption that volatility
is constant, as in the standard Black--Scholes model,  is not realistic.
Practically, this reflects itself in an implied volatility that depends on the 
pricing parameters.  This means that, in order to  match observed prices, the 
volatility that one  needs to use  in the   Black--Scholes option pricing formula   depends on  
time to maturity and log-moneyness, with moneyness being the  strike
price over the current price of the underlying.  The implied volatility
is a convenient way to parameterize the price of a financial contract
relative to a particular underlying. It gives insight about how the market
deviates from the ideal Black--Scholes situation.
After calibration of an implied volatility model to  liquid contracts,
this model can be used for pricing less liquid contracts 
written on the same underlying.
It is, therefore, of interest to identify a 
consistent parameterization of the implied volatility that  
corresponds  to an underlying model for stochastic volatility fluctuations.
As in \cite{sv1}, a main objective is to construct a stochastic volatility model that is
a stationary process and that makes it possible to consider general times to maturity. 
  For background on stochastic volatility models, we refer the reader to the books and surveys by
\citet{fouque11,gatheralbook,ghysels,gul,henry,rebonato} (see the references therein).
We also refer the reader to our paper on fractional stochastic volatility, \cite{sv1},
for further references on the recent literature on the class of volatility models we consider 
here. 

Empirical studies suggest that volatility may 
exhibit  a ``multi-scale'' character with long-range correlations, as in
\cite{bollerslev,breidt,viens2,cont0,cont,engle,oh}. That means that
correlations decay as a power law in time offset,
while they would decay as an exponential function if
stochastic volatility were Markovian. 
Here we  seek to identify 
parametric forms for the implied volatility  consistent with
such long-range correlations.  In our recent paper \cite{sv1},
we considered this question within the context where the magnitude of the 
volatility fluctuations is small. Here, we consider the situation where 
the magnitude of the volatility fluctuations is of the same order as the mean volatility.
Indeed,  empirical studies show that the volatility fluctuations may be quite 
large: \cite{breidt,cont0,engle}.  
 While in   \cite{sv1} the volatility fluctuations
were small, leading to a (regular) perturbative situation,
here the situation  is  different in that it is the fast mean reversion
(fast relative to the diffusion time of the underlying) that allows  us to push through an asymptotic analysis.
The presence of long-range correlations in this context gives
a novel, singular perturbation situation. The analysis becomes  significantly
more complex. In particular, the detailed  analysis of the covariation process is
an important ingredient.  We consider here option pricing, but the approach set forth
is general and will be useful in other financial contexts as well.

 It follows from our analysis that the form of the
implied volatility surface is similar to the one obtained in the Markovian case.
This confirms the robustness of the implied volatility parametric
 model with respect to the underlying price dynamics. There are, however,
 central differences. In particular, the 
 long-range correlations produce a volatility covariance
 that is not integrable, which in turn gives  an
 implied volatility surface  that is  a {\it random field},
 whose statistics can be described in detail. 
Moreover, in the long-range case, the implied volatility has a fractional
 behavior as a function of time to maturity. 
The empirical study in \cite{25} shows that, in order to fit well the 
implied volatility, it is appropriate to consider a two-time scale
model with one slow and one fast volatility factor.  In \cite{sv1},
we considered a slow factor, which is closely associated with a small
fluctuation factor. Here, we consider a fast factor with large
fluctuations. Taken together, we have a generalization
of the two-factor model of \cite{25,fouque11} for processes with long-range correlations. 
This leads to a fractional term structure of the implied volatility.
It was shown in \cite{22} that such a term structure may be useful for fitting
the implied volatility under certain market conditions.

\subsection*{Long Memory and Fast Mean Reversion}    
    
As mentioned above, the asymptotic regime considered in this paper
is the situation where the  volatility is fast mean reverting.
We denote its time scale by $\eps$, the small parameter 
in our model. The volatility then decorrelates on the time scale $\eps$. 

Stochastic volatility models are most often set with  a volatility driving process
that  has mean zero and mixing properties. This means that the random values of the volatility driving process at times
$t$ and $t+\Delta t$, which are $Z_t^\eps$ and  $Z_{t+\Delta t}^\eps,$   become rapidly uncorrelated when 
$\Delta t \to \infty$, i.e.,
the autocovariance function 
 ${\cal C}^\eps (\Delta t)=\EE[Z_t^\eps Z_{t+\Delta t}^\eps   ]$  decays rapidly  to zero 
as $\Delta  t \to \infty$.
More precisely, we say that the volatility driving process  is mixing if its autocovariance function 
decays fast enough at infinity, so that it is absolutely integrable
\begin{eqnarray}
\int_0^\infty |{\cal C}^\eps(t)| dt  < \infty \, .
\end{eqnarray}
In this case, we may associate the process with the finite correlation time 
  $t_c= 2 \int_0^\infty  {\cal C}^\eps(t) dt / {\cal C}^\eps(0)$, which is of order $\eps$.
 
Stochastic volatility models with long-range correlation properties
have recently attracted a lot of attention,
as more and more data collected under various situations confirm
that this situation can be encountered in many different markets.
Qualitatively, the long-range correlation property means that the random process has
long memory (in contrast with a mixing process).
This means that the correlation between the random values $Z_t^\eps$ and 
$Z_{t+\Delta t}^\eps$ taken at two times separated by $\Delta t$ is not completely negligible 
even for large  $\Delta t$.
More precisely, we say that the random process $Z_t^\eps$ 
 has the $H$-long-range correlation property if its  autocovariance function satisfies
\begin{eqnarray}
\label{eq:decayphi0}
{\cal C}^\eps( t ) \stackrel{|t| \to \infty}
\simeq r_H \Big|\frac{t}{\eps}\Big|^{2H-2}
  \, , 
\end{eqnarray}
where $r_H >0$ and $H\in (1/2,1)$.
We refer to $H$ as the Hurst exponent.
Here the correlation time   $\eps$ is the 
critical time scale beyond  which the power law  behavior (\ref{eq:decayphi0}) is  valid. 
Note that the autocovariance function is not integrable as $2H-2\in (-1,0)$,
which means that a random process with the $H$-long-range correlation property is not  mixing.
As we describe in more detail below, a common approach for modeling
long-range dependence is by using fractional Brownian motion (fBm) processes
as introduced in \cite{mandelbrot}.
    
Long-memory stochastic volatility models are easy to introduce, but difficult to 
analyse. This is  largely due to the fact that the 
volatility process is neither a Markov process nor a semimartingale.
It is, however, important to note that the price process is still
a semimartingale, and the problem formulation does not entail
arbitrage (\cite{mendes1}), as has been argued for some models whose price 
process itself is driven by fractional processes, as in \cite{bjork,roger,shiryaev}.
A main motivation for long-memory is to be able to fit observed implied volatilities.
   One common challenge
  regarding the fitting of implied volatility surfaces is to capture a  strong moneyness
  dependence for short time to maturity without creating artificial behavior for long time to maturity.
  Another typical challenge is to retain a strong parametric dependence for long maturities 
  despite averaging effects that occur in this regime, as discussed in \cite{boll2,bollerslev,coutin,sun}.
  We remark that models involving jumps have been promoted as one 
  approach to meet these challenges by \cite{carr,tankov}.    Recent works show  that stochastic volatility models with long-range
  dependence also   provide a promising  framework for meeting such challenges. 
Approaches  based on using  fractional noises in the description of the stochastic volatility
process were used by \cite{comte,coutin}. 
Such stochastic volatility models with long-range dependence
can capture the steepness of long-term volatility
smiles without overemphasizing  the short-run persistence. 
In order to get explicit results for the implied volatility, a number of asymptotic
regimes have been considered. Chief among them has been the regime
of short time to maturity.
    The   model presented in \cite{coutin} was recently revisited in \cite{jacquier},
     where
     short and long time to maturity asymptotics are analysed using large deviations theory.  
In   \cite{alos07}, the authors   use Malliavin calculus  
to  decompose the option price  as the sum of the classic
Black--Scholes formula price
and a term due to the volatility of the volatility.
In the Black--Scholes formula, they use a 
volatility parameter that is equal to the root-mean-square
future average volatility plus a term due to the leverage effect (i.e., the correlation between the underlying return and its changes in volatility).
  Their model is a fractional version of the Bates model (\cite{bates}). 
  They find that the implied volatility flattens in the long-range dependent case in the limit of short time to maturity.
In \cite{forde}, the authors   use large deviation principles to compute the short time to maturity 
asymptotic form of the implied volatility.
They consider the leverage effect and obtain results that are consistent with those in \cite{alos07}.
They consider stochastic volatility models driven by fBms,
which are analysed using rough path  theory.  They also consider 
long time asymptotics for some fractional processes.
Short-time-to-maturity  asymptotic results were also recently presented in   \cite{viens3}  in a context of long-range processes.   
In \cite{gatheral2},  the authors 
   consider the rough Bergomi model, or ``rBergomi'' model, and discuss
   the form of the associated implied volatility term structure.
In  \cite{fukasawa11},  the author discusses how small volatility fluctuations
with long-range dependence impact  the implied volatility,
   as an application of the general theory he sets forth.  
In that paper, as well as in \cite{alos07}, the authors use a model where time 0 plays
   a special role, and hence the modeling is not completely satisfactory, because it leads
   to a non-stationary volatility model.  
      On the other hand, in  \cite{sv1}, which deals with small
   volatility fluctuations, the authors use a formulation with a stationary model. 
This is also the case in the recent paper by \cite{fukasawa17}, which considers
   short time asymptotics in the rough volatility case, with $H<1/2$. 
This distinction is important: 
if  the volatility factor is a fBm emanating from the origin, 
then the  implied volatility surface is identified conditioned  on the present value of the volatility factor only. 
In our paper, we use a stationary model so that the implied volatility surface depends on the
   path of the volatility factor until  the present, reflecting the non-Markovian nature
   of fBm.  
   We discuss in detail in Section \ref{sec:tT} the consequences of
   this for  interpretating the implied volatility surface as a random field. 
   Recently,  pricing approximations in the regime of 
   small fractional volatility fluctuations were  presented in \cite{alos2}.
  In terms of computation of prices for general maturities and fractional volatility fluctuations, so far 
mainly numerical approximations have been available.  Here
we present an asymptotic regime based on fast mean reversion which
 gives explicit price approximations.  
Together, the results of \cite{sv1} and the current paper
make it possible to construct a fractional, two-time-scale stochastic volatility model,
which gives enough flexibility to fit both the short and long time to maturity parts of the implied volatility surface.

Let us note that we consider here the long-range correlation case where $H>1/2$ 
 as opposed to the rough volatility case where $H<1/2$. 
 Indeed, both regimes have been identified from the empirical perspective. We refer the reader, for instance, to
 \cite{gatheral1} for observations of rough volatility, 
and to \cite{viens1} for cases of  long-range volatility.
 A long-range, mean-reverting volatility situation is  
 reported in \cite{jensen} in a discrete modeling framework.  
 Long-range volatility situations are also reported 
for currencies in \cite{curr},  for commodities in \cite{energy}, and for equity indices in
\cite{mal}. 
Analysis of electricity markets data typically gives  $H<1/2$, as reported in
\cite{simonsen,rypdal,bennedsen}. 
  We  believe  that both the rough and the long-range cases   are important and can be observed 
depending on the specific market and regime. 
 Even though the ``rough'' case 
with $H <1/2$ may be the  most common situation, the understanding of the situation where $H>1/2$
 may be of particular importance for pricing and hedging.  
In this paper, we only consider the analytic aspects of our model. The fitting with
  respect to specific data is beyond the scope of this paper and will be presented 
in future work.

  The fractional model we set forth here   produces  typical  ``stylized facts'', 
such as heavy tails of returns, volatility clustering, 
   mean reversion, and long memory or volatility persistence. 
   Additionally, here we incorporate the leverage effect.  This term 
   was coined by
   \cite{black},  referring to  stock-price movements that are
   correlated (typically negatively) with volatility,
   as  falling stock prices may imply  more uncertainty, and hence volatility.   
   Note, however, that the model for the implied volatility surface derived below is
  linear in log-moneyness. This may seem somewhat restrictive from the point
  of view of fitting, because, in many cases, a strong skew 
  in log-moneyness  may  be observed in certain markets. 
  This has  particularly been the case for  stock markets, but relatively less
  so in other markets, such as fixed income markets.
Nevertheless, if one considers higher order
  approximations, then they also generate skew effects. 
  A number of modeling issues not considered here,
such as transaction costs, bid-ask spreads and liquidity, 
  may also affect the skew shape.
   Note also that,
    for simplicity, we do not incorporate a non-zero interest rate or a market price for risk aspects.

\subsection*{ Rapid-Clustering,  Long-Memory and the Implied Surface}

Next, we will summarize the main result of the paper from the point of view 
of calibration, that is, the form of the implied volatility surface in the context
of a stochastic volatility modeled by a fast process with long-range correlation properties.
We will first summarize   some aspects of the modeling.

We consider a continuous-time stochastic volatility model that is a smooth function of a Gaussian 
long-range process. 
 Explicitly, we model the fractional stochastic volatility (fSV) as a smooth function of a fractional 
Ornstein--Uhlenbeck  (fOU) process.  
The fOU process
is a classic model for a stationary  process with a fractional  
long-range correlation structure.  This process can be expressed in terms of
an integral of a  fBm  process. The distribution of a fBm 
 process is characterized in terms of the Hurst exponent $H \in (0,1)$.  
The fBm process is 
locally H\"older continuous of exponent $H'$  for all 
$H'  < H$, and this property is inherited by the 
fOU process.  
The fBm process, $W_t^H$,  is also self-similar in that
\begin{equation}
 \left\{
 W^H_{\alpha t} , t\in \RR  
 \right\}
 \stackrel{dist.}{=}
  \left\{
 \alpha^H W^H_{t} , t\in \RR 
 \right\}
~ \hbox{for all}~   \alpha > 0.
\end{equation}
The self-similarity property is inherited approximately by  the fOU
process on time scales smaller than the mean-reversion time
of the fOU process, which  we denote by $\eps$ below. In this sense, we
may refer to the fOU process as a multi-scale process on short time scales. 
The case  $H \in (1/2,1)$  that we address in this paper gives  
a fOU process that is   long-range.
This regime corresponds to  a persistent
process where consecutive increments of the fBm
are positively correlated.
The stronger, positive correlation for the consecutive increments
of the associated fBm process with increasing $H$ values gives a 
smoother process whose autocovariance function decays slowly. 
For more details regarding the fBm and fOU processes,
we refer the reader to \cite{oksendal,coutinb,taqqu,mandelbrot,cheridito03,kaarakka}. 
  
The volatility driving process is
the $\eps$-scaled fOU process defined by 
\begin{equation}\label{eq:OU0}
Z^\eps_t = \eps^{-H} \int_{-\infty}^t e^{-\frac{t-s}{\eps}} dW^H_s   .
\end{equation}
It is a zero-mean, stationary Gaussian process that exhibits long-range correlations
for the Hurst exponent $H \in (1/2,1)$. It is important to note that this is a process
whose ``natural time scale'' is $\eps$,   in the sense that the mean-reversion time, or
time before the process reaches its equilibrium distribution, is of the order of $\eps$.
It is also important  to note that the decay of the correlations (on the $\eps$ time scale) is
polynomial rather than exponential, as in the standard Ornstein--Uhlenbeck process.
Explicitly, the correlation  of the process between times $t$ and $t+\Delta t$
decays as $(\Delta t/\eps)^{2H-2}$, while 
 the variance of the process is independent of $\eps$.

In this paper, we consider a stochastic volatility model that is a smooth function of the 
rapidly varying fOU process with Hurst coefficient $H \in (1/2,1)$.
It is given by
\begin{equation}
\label{def:stochmodel0}
\sigma_t^\eps = F( Z_t^\eps) ,
\end{equation}
where $F$ is a smooth, positive, one-to-one bounded function with bounded derivatives, 
and with an additional technical condition that is  given in Eq.  (\ref{eq:cond}). 
The process $\sigma_t^\eps$ inherits the long-range correlation properties of the fOU $Z^\eps_t$.

The main result, in Section \ref{sec:implied},  is 
an expression for 
the implied volatility of the 
European Call Option for strike $K$, maturity $T$, and current time $t$,
 \begin{equation}\label{eq:iv}
I_t = \EE\Big[ \frac{1}{T-t} \int_t^T (\sigma_s^\eps)^2 ds \big| {\cal F}_t \Big]^{1/2} 
+  \overline\sigma  a_F
\left[  \left( \frac{\tau}{\bar\tau} \right)^{H-1/2} +  \left( \frac{\tau}{\bar\tau} \right)^{H-3/2} {\log\Big(\frac{K}{X_t}\Big)} \right]
  .
\end{equation}
Here
\begin{equation}
a_F = 
 \eps^{1-H} \frac{ \widetilde{\sigma} 
 \rho 
  \left<FF'\right> \bar\tau^{H} }{ 2^{3/2} \overline{\sigma} \Gamma(H+{3}/{2})}  ,
\end{equation}
$\tau=T-t$ is time to maturity,
$\rho$ the correlation between the Brownian motion
driving the fBm  and the Brownian
motion driving the underlying,
and 
\begin{equation}
\label{def:bartau}
 \bar\tau = \frac{2}{{\overline\sigma}^2}   
\end{equation}
is   the characteristic  diffusion time.
Furthermore, we have 
 \begin{eqnarray*}
\overline{\sigma}^2 =  \left< F^2\right> &=& \int_\RR F(\sigma_{{\rm ou}} z )^2 p(z) dz, \\
 \widetilde{\sigma}  =  \left< F\right> &=& \int_\RR F(\sigma_{{\rm ou}} z ) p(z) dz, \\
\left<FF'\right> &=&   \int_\RR F(\sigma_{{\rm ou}} z ) F'(\sigma_{{\rm ou}} z ) p(z) dz ,
\end{eqnarray*}
where $ \sigma^2_{{\rm ou}}=1/(2 \sin(\pi H))$ and $p(z)$ is the probability density function (pdf) of the standard normal distribution. 
In other words, we form moments of the
volatility function averaged with respect to the invariant distribution of the 
fOU process $Z_t^\eps$. 

The first term in Eq. (\ref{eq:iv}) is indeed the expected effective volatility until maturity
conditioned on the present. The second term is a skewness term that is non-zero  only when
the volatility process and the underlying are correlated so that $\rho$ is non-zero.
Note that the exponent of the fractional term structure depends on the
Hurst exponent, which determines the smoothness and the decorrelation
rate of the volatility driving process $Z_t^\eps$. 
The smoother the process, the larger
the implied volatility for long times to maturity.

In the fast case presented here with large and fast  
volatility fluctuations, the  implied volatility explodes  in the regime
of short time to maturity.
 Indeed, short time to maturity means that the time to maturity is smaller than the diffusion time (\ref{def:bartau}),
but larger than the mean-reversion time $\eps$. 
Therefore, short time to maturity involves 
large volatility fluctuations over a short maturity horizon  
resulting in a moneyness  correction that explodes and  dominates the pure maturity 
term.
 In the context of short or long times to maturity,   the conditional
 expected effective volatility  gives a small contribution,
 and we have for short times to maturity and $K \neq X_t$
   \begin{equation}\label{eq:iv2}
I_t \sim
\overline{\sigma} 
a_F  \Big[     \left( \frac{\tau}{\bar\tau} \right)^{H-3/2} {\log\Big(\frac{K}{X_t}\Big)} \Big]
  ,
\end{equation} 
and for long times to maturity   
  \begin{equation}\label{eq:iv3} 
  I_t \sim
\overline{\sigma} 
a_F    \left( \frac{\tau}{\bar\tau} \right)^{H-1/2}  
     .
    \end{equation}

 We note here that the fractional scaling in the skewness term in  Eq. (\ref{eq:iv})
is exactly the fractional scaling that corresponds to the case of long
time to maturity and small volatility fluctuations given in  
\cite{sv1}. That means that, with long times to maturity, we have a situation reminiscent 
of the one we have here with rapid volatility fluctuations. Here, however,  
the volatility fluctuations are large compared to the small volatility
fluctuations considered in \cite{sv1}.

We remark also that the case with a mixing volatility, and hence integrable correlation
function for the volatility fluctuations, would correspond to $H \searrow 1/2$.  Note, however, that our derivation
is valid only for $H\in (1/2,1)$.
If we consider the formula (\ref{def:sigmaphi}) for $\sigma_\phi$ that determines
the variance of the first term in Eq.~(\ref{eq:iv}), we observe that it vanishes when $H \searrow 1/2$,
which shows that  the first term in Eq.~(\ref{eq:iv}) becomes deterministic.
In the mixing case, the first-order correction to
 the implied volatility is deterministic,  
while  the non-integrability 
of the volatility covariance function makes it a stochastic
process in the general, long-range case  with a variance that goes to zero 
as $H \searrow 1/2$.
 Indeed, in the limit case $H \searrow 1/2$,
we get a result similar to the one obtained in \cite[Section 5.2.5]{fouque00} that deals with the mixing case.
Explicitly, we consider the mixing case where the volatility driving process is an ordinary Ornstein--Uhlenbeck process;
moreover, the interest rate and market price of volatility risk are zero. Then 
 \cite[Eq. (5.55)]{fouque00} gives the implied volatility in terms of a coefficient $V_3$ defined in  \cite[Section 5.2.5]{fouque00},
 \begin{equation}
 \label{eq:iv4}
I_t = \overline{\sigma} 
- V_3  \Big[    \frac{1}{ 2 \overline{\sigma} }  + 
     \frac{1}{ \overline{\sigma}^3 \tau}  \log\Big(\frac{K}{X_t}\Big) 
     \Big]
   ,
\end{equation}
which has the same form as the formal limit of (\ref{eq:iv}) as $H \searrow 1/2$.
The averaging expression giving the coefficient $V_3$ does not, however, correspond
to the interpretation we arrive at here by the formal limit $H \searrow 1/2$. That is because
the singular perturbation situation we consider is  in fact  ``singular'' at $H=1/2$, and 
the ordering of important terms changes. Nevertheless,  it is important
from the calibration point of view  that we have continuity of the implied volatility
parameterization and its form at $H=1/2$,  providing  robustness to  the 
asymptotic framework. 

In Section \ref{sec:tT}, we give the complete statistical description
of  the stochastic correction coefficient,
which determines the random component of the price correction and the implied volatility (the first term in Eq.~(\ref{eq:iv})). 
It is a random function of the maturity $T$  and the current time $t$ with Gaussian statistics
and with a covariance function that we describe in detail.
This covariance function has interesting and non-trivial, self-similar properties,
and this function is important in order to   construct and characterize estimators of the implied volatility surface.

 \section*{Outline}   
 The outline of the paper is as follows. In Section \ref{sec:fou}, 
 we describe the fractional Ornstein--Uhlenbeck process and derive
 some fundamental {\it a priori} bounds. 
 In Section \ref{sec:svm}, we describe the stochastic volatility model.
In Section~\ref{sec:option}, we derive
 the expression for the price in the fast mean-reverting fractional case.
 The derivation is based on the martingale method. That is, we make an ansatz 
 for the price as a process that has the correct payoff and whose leading-order term is a martingale. 
 Then indeed this process is the leading-order
 expression for the price with an error that is of the order of the non-martingale
 part. This approach  involves introducing correctors so that the non-martingale part
 is pushed to a  small term;  we give the resulting decomposition  in Section
 \ref{sec:option}. Based on the expression for the price,  we derive  the associated
 implied volatility in Section~\ref{sec:implied} and present our concluding
 remarks in Section~\ref{sec:conclusion}.  We give a convenient Hermite decomposition
of the volatility in Appendix~\ref{sec:hermite}.  A number of the technical lemmas are proved
 in   Appendix~\ref{sec:app}. 

\section{The Rapid Fractional Ornstein--Uhlenbeck Process}
\label{sec:fou}
We use a rapid fractional Ornstein--Uhlenbeck (fOU) process as the volatility factor and
describe here how this process can be represented in terms of a fractional Brownian motion.
Because fractional Brownian motion can be expressed in terms of ordinary Brownian motion,
we also arrive at an expression for the rapid fOU process as a filtered version
of Brownian motion.

A fractional Brownian motion (fBm) is a zero-mean Gaussian process $(W^H_t)_{t\in \RR}$  
with the covariance
\begin{equation}
\label{eq:covfBM}
\EE[ W^H_t W^H_s ] = \frac{\sigma^2_H}{2} \big( |t|^{2H} + |s|^{2H} - |t-s|^{2H} \big) ,
\end{equation}
where $\sigma_H$ is a positive constant.
We use the following moving-average stochastic integral representation of the
fBm (\cite{mandelbrot})
\begin{equation}
\label{mandelbrot}
W^H_t = \frac{1}{\Gamma(H+\frac{1}{2})} 
\int_{\RR} \Big( (t-s)_+^{H - \frac{1}{2}} -(-s)_+^{H - \frac{1}{2}} \Big) dW_s ,
\end{equation}
where $(W_t)_{t \in \RR}$ is a standard Brownian motion over $\RR$.
Then $(W^H_t)_{t\in \RR}$ is indeed a zero-mean Gaussian process with
the covariance (\ref{eq:covfBM}), and we have
\begin{eqnarray}
\nonumber
\sigma^2_H &=& \frac{1}{\Gamma(H+\frac{1}{2})^2} 
\Big[ \int_0^\infty \big( (1+s)^{H - \frac{1}{2}} -s^{H - \frac{1}{2}} \big)^ 2 ds +\frac{1}{2H}\Big] \\
&=&  \frac{1}{\Gamma(2H+1) \sin(\pi H)}.
\end{eqnarray}

We introduce the $\eps$-scaled fOU as
\begin{equation}\label{eq:fOU}
Z^\eps_t = \eps^{-H} \int_{-\infty}^t e^{-\frac{t-s}{\eps}} dW^H_s  = 
\eps^{-H}  W^H_t - \eps^{-1-H}  \int_{-\infty}^t e^{-\frac{t-s}{\eps}} W^H_s ds.
\end{equation}
Thus, the fOU process  is, in fact, a fractional Brownian motion with a restoring force towards zero.
It is a zero-mean, stationary Gaussian process,
with variance
\begin{equation}
\label{eq:sou}
\EE [ (Z^\eps_t)^2 ] = \sigma^2_{{\rm ou}}, \mbox{ with } \sigma^2_{{\rm ou}} = \frac{1}{2} \Gamma(2H+1)  \sigma^2_H = \frac{1}{2\sin (\pi H)} ,
\end{equation}
which is independent of $\eps$,
and covariance
\begin{eqnarray*}
\EE [ Z^\eps_t Z^\eps_{t+s}  ] &=&\sigma^2_{{\rm ou}} {\cal C}_Z\Big(\frac{s}{\eps}\Big) ,
\end{eqnarray*}
which is a function of $s/\eps$ only, with
\begin{eqnarray*}
{\cal C}_Z(s) &=& 
 \frac{1}{\Gamma (2H+1)}
 \Big[ \frac{1}{2} \int_\RR e^{- |v|}| s+v|^{2H} dv - |s|^{2H}\Big]
\\
&=&
\frac{2 \sin (\pi H)}{\pi}
\int_0^\infty \cos (s x  ) \frac{x^{1-2H}}{1+x^2} dx .
\end{eqnarray*}
This shows that $\eps$ is the natural scale of variation of the fOU $Z^\eps_t$.
Note that the random process $Z^\eps_t$ is neither a martingale, nor a Markov process.
For $H \in  (1/2,1)$, it possesses long-range correlation properties
\begin{equation}
\label{eq:corZG3}
{\cal C}_Z(s) =  
 \frac{1}{  \Gamma(2H-1)}  s^{2H-2}
+ o\big( s^{2H-2}\big) 
, \quad \quad s \gg 1.
\end{equation}
This shows that the correlation function is non-integrable at infinity.
In this paper, we focus on the case $H \in (1/2,1)$.
 
We remark that if
$H=1/2$, then the standard Ornstein--Uhlenbeck process (synthesized  with a standard Brownian motion) is a stationary 
Gaussian Markov process with an exponential correlation, and hence a mixing process.
It is possible to simulate paths of the  fOU process using the Cholesky method (see Figure \ref{fig1a}),
or other well-known methods described in \cite{6,bardet}.

\begin{figure}
\begin{center}
\includegraphics[totalheight=6.0cm]{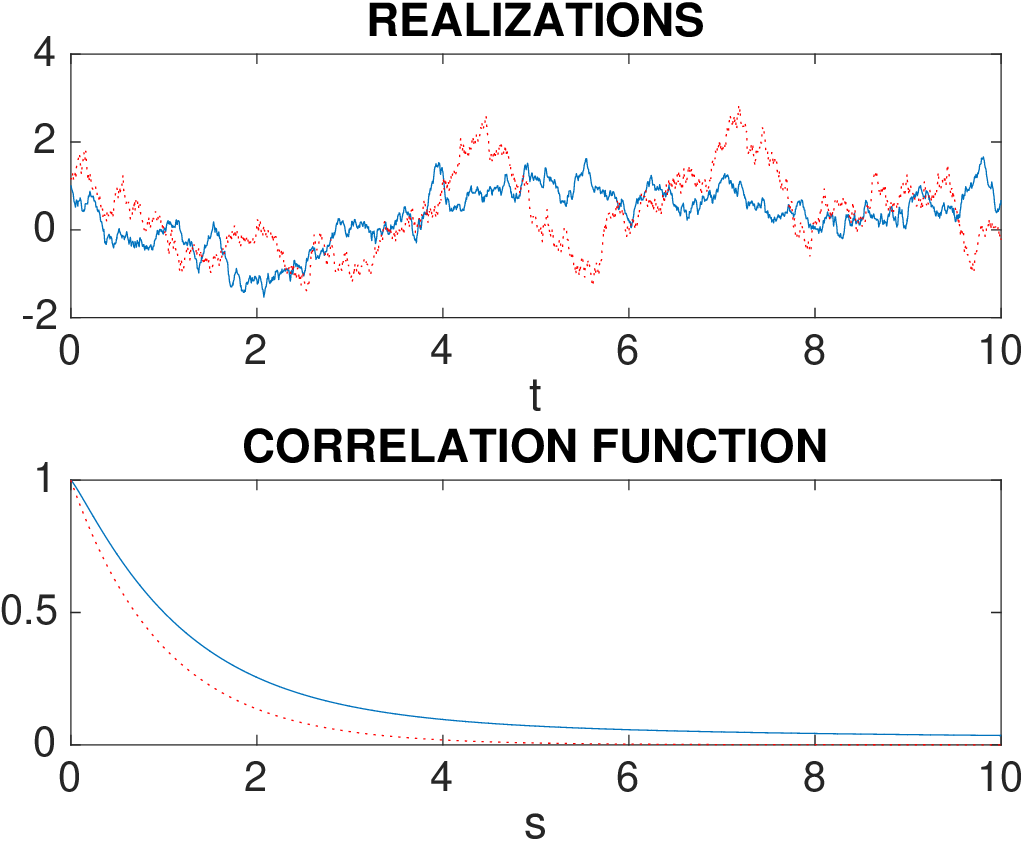}
\end{center}
    \caption{The top plot shows a realization, $Z_t^\eps$, $t \in (0,10)$,  of the fOU process with Hurst index $H=0.6$ and correlation time 
    $\eps=1$ (blue   solid line) and  a realization of the  standard
    Ornstein--Uhlenbeck process with $H=1/2$ and $\eps=1$ (red dashed line).   
The trajectories are more regular when $H$ is larger.
The bottom plot shows the corresponding correlation functions, ${\cal C}_Z(s)$, and the ``heavy'' tail of the blue solid line
of the case $H=0.6$ gives the long-range property. 
 }
\label{fig1a}
\end{figure}

Using Eqs. (\ref{mandelbrot}) and (\ref{eq:fOU}), we arrive at
the moving-average integral representation of the scaled fOU as
\begin{equation}
Z^\eps_t = \sigma_{{\rm ou}} \int_{-\infty}^t {\cal K}^\eps(t-s) dW_s,
\end{equation}
where 
\begin{equation}\label{eq:kdef}
{\cal K}^\eps(t) = \frac{1}{\sqrt{\eps}} {\cal K}\Big(\frac{t}{\eps}\Big),\quad \quad
{\cal K}(t) =\frac{1}{\Gamma(H+\frac{1}{2})
{\sigma_{\rm ou}}
} 
 \Big[ t^{H - \frac{1}{2}} - \int_0^t (t-s)^{H - \frac{1}{2}} e^{-s} ds \Big] .
\end{equation}
The main properties of the kernel ${\cal K}$ in our context are the following   (valid for any $H \in (1/2,1)$):\\
- ${\cal K}$ is {nonnegative-valued,}
${\cal K} \in L^2(0,\infty)$ with $\int_0^\infty {\cal K}^2(u) du = 1$, 
but ${\cal K} \not\in L^1(0,\infty)$,\\
- for short times $t \ll 1$
\begin{equation}
{\cal K} (t) = \frac{1}{\Gamma(H+\frac{1}{2}) 
{\sigma_{\rm ou}}
 } 
\Big( t^{H - \frac{1}{2}}  + O\big(  t^{H +\frac{1}{2}} \big) \Big) ,
\end{equation}
- for long times $t \gg 1$
\begin{equation}
{\cal K} (t) = \frac{1}{\Gamma(H-\frac{1}{2})
{\sigma_{\rm ou}}
}  
\Big( t^{H - \frac{3}{2}}  +O\big(  t^{H - \frac{5}{2}} \big) \Big) ,
\end{equation}
and, in particular, ${\cal K} (t)- \frac{1}{
{\sigma_{\rm ou}}
\Gamma(H-\frac{1}{2})}   t^{H - \frac{3}{2}}$ $\in$  $L^1(0,\infty)$.

\section{The Stochastic Volatility Model}
\label{sec:svm}
The price of the risky asset follows the stochastic differential equation
\begin{equation}
dX_t = \sigma_t^\eps X_t dW^*_t ,
\end{equation}
where the stochastic volatility is
\begin{equation}
\label{def:stochmodel}
\sigma_t^\eps = F(Z_t^\eps) ,
\end{equation}
and  $Z_t^\eps$ is the scaled fOU introduced in the previous section,
which is  adapted to the Brownian motion $W_t$. Moreover,
 $W^*_t$ is a Brownian motion that is correlated to the stochastic volatility through
\begin{equation}\label{eq:corr} 
W^*_t = \rho W_t + \sqrt{1-\rho^2} B_t ,
\end{equation}
where the Brownian motion $B_t$ is independent of $W_t$.

The function $F$ is assumed to be one-to-one, positive-valued,
smooth, bounded and with bounded derivatives.
Accordingly, the filtration ${\cal F}_t $ generated
by $(B_t,W_t)$ is also the one generated by $X_t$.
Indeed, it is equivalent to the filtration generated by $(W^*_t,W_t)$, or $(W^*_t,Z^\eps_t)$.
Because $F$ is one-to-one, it is equivalent to the filtration generated by $(W^*_t,\sigma_t)$.
Because $F$ is positive-valued,
it is equivalent to the filtration generated by $(W^*_t,(\sigma_t^\eps)^2)$, or $X_t$. 

We denote the Hermite coefficients of the volatility function  $F$
with respect to the invariant distribution of the fOU process by $C_k$, 
\begin{equation}
C_k = \int_\RR H_k(z) F^2(\sigma_{{\rm ou}} z) p(z) dz ,\quad \quad H_k(z) = (-1)^k e^{z^2/2} \frac{d^k}{dz^k} e^{-z^2/2},
\end{equation}
with $p(z)=\exp(-z^2/2)/\sqrt{2\pi}$.
 We use these in Appendix \ref{sec:hermite}  to derive some technical lemmas.  
Indeed, there is a technical reason requiring that $F$ satisfies the following condition:
there exists some $\alpha>2$ such that
\begin{equation}\label{eq:cond}
\sum_{k=0}^\infty \frac{\alpha^k C_k^2}{k!} < \infty .
\end{equation}
 
As discussed above, the volatility driving process  $Z_t^\eps$ possesses 
long-range correlation properties. As we now show, the volatility process $\sigma_t^\eps$ itself  inherits this property.

 \begin{lemma}
\label{prop:process}
We denote, for $j=1,2$,
\begin{equation}
\label{def:meanF}
\left< F^j \right> =  \int_{\RR}  F(\sigma_{{\rm ou}}z)^j p(z) dz,
\quad \quad 
\left< {F'}^j \right> =  \int_{\RR}  F'(\sigma_{{\rm ou}}z)^j p(z) dz,
\end{equation}
where $p(z)$ is the pdf of the standard normal distribution.

\begin{enumerate}
\item
The process $\sigma_t^\eps$ is a
stationary random process with mean $\EE[\sigma_t^\eps] = \left<F\right>$ and variance $
{\rm Var}(\sigma_t^\eps) = \left<F^2\right>-\left<F\right>^2$, independently of $\eps$.

\item
The covariance function of the process $\sigma_t^\eps$ is of the form
\begin{eqnarray}
\label{eq:corrY1}
{\rm Cov}\big( \sigma_t^\eps , \sigma_{t+s}^\eps \big) &=& \big( \left<F^2\right>-\left<F\right>^2\big) {\cal C}_\sigma \Big(\frac{s}{\eps}\Big),
\end{eqnarray}
where the correlation function ${\cal C}_\sigma$ satisfies ${\cal C}_\sigma(0)=1$ and
\begin{eqnarray}
{\cal C}_\sigma(s) &=&
 \frac{1   }{  \Gamma(2H-1)}
\frac{ \sigma_{{\rm ou}}^2\left< F' \right>^2 }{ \left< F^2\right>-\left<F\right>^2}
s^{2H-2} +o\big( s^{2H-2}  \big)  ,
\quad \mbox{ for  } s \gg 1.
\label{eq:corrY12}
\end{eqnarray}
\end{enumerate}
\end{lemma}
Consequently,  the process $\sigma_t^\eps$ possesses long-range correlation properties (i.e. its correlation function 
is not integrable at infinity).

\begin{proof}
  The fact that $\sigma_t^\eps$ is a stationary random process with mean
  $\left< F \right>$ is straightforward in view of the definition (\ref{def:stochmodel})
  of $\sigma_t^\eps$.  
 
For any $t,s$, the vector $\sigma_{{\rm ou}}^{-1} (Z^\eps_t ,Z^\eps_{t+s})$ is a Gaussian
random vector with mean $(0,0)$ and $2\times 2$ covariance matrix
$$
{\bf C}^\eps =  \left( \begin{array}{cc}
1 &{\cal C}_Z(s/\eps) \\
{\cal C}_Z(s/\eps) & 1 \end{array} \right).
$$
Therefore, denoting ${F}_c(z) = F(\sigma_{{\rm ou}}z) - \left< F \right>$,
the covariance function of the process $\sigma_t^\eps$ is
\begin{eqnarray*}
{\rm Cov} (\sigma_t^\eps , \sigma_{t+s}^\eps) 
&=& \EE \big[{F}_c(\sigma_{{\rm ou}}^{-1} Z_t^\eps)  {F}_c(\sigma_{{\rm ou}}^{-1} Z_{t+s}^\eps) \big]
\\
&=&
\frac{1}{2 \pi  \sqrt{ \det {\bf C}^\eps}}
\iint_{\RR^2} {F}_c(z_1) {F}_c(z_2) 
\exp\Big( - \frac{(z_1 ,z_2) {{\bf C}^\eps}^{-1} (z_1 ,z_2)^T}{2 }\Big) dz_1 dz_2 
\\
&=& 
\Psi\Big( {\cal C}_Z \Big(\frac{s}{\eps}\Big)\Big) ,
\end{eqnarray*}
with
\begin{eqnarray*}
\Psi(C) = 
\frac{1}{2 \pi \sqrt{1- C^2} }
\iint_{\RR^2}  {F}_c(z_1) {F}_c(z_2) \exp
\Big( - \frac{z_1^2+z_2^2 -2  C  z_1 z_2}{2 (1 -
C^2 )} \Big) dz_1 dz_2 \, .
\end{eqnarray*}
This shows that ${\rm Cov} (\sigma_t^\eps , \sigma_{t+s}^\eps) $ is a function of $s/\eps$ only.
Moreover, the function $\Psi$ can be expanded in powers of $C$ for small $C$,
\begin{eqnarray*}
\Psi(C) &=& 
\frac{1}{2\pi  }
\iint_{\RR^2}  {F}_c(z_1) {F}_c(z_2) \exp
\Big( - \frac{z_1^2+z_2^2 }{2  }\Big) dz_1 dz_2
\\
&&+ C 
\frac{1}{2\pi }
\iint_{\RR^2}  z_1 z_2 {F}_c(z_1) {F}_c(z_2) \exp
\Big( - \frac{z_1^2+z_2^2 }{2 }\Big) dz_1 dz_2 +O(C^2), \quad \quad C \ll 1,
\end{eqnarray*}
which gives with (\ref{eq:corZG3}) the form (\ref{eq:corrY12}) of the 
correlation function for  $\sigma_t^\eps$.
\end{proof}

\section{The Option Price}
\label{sec:option}
Our aim is to  compute the option price defined as the martingale
\begin{equation}
M_t =\EE\big[ h(X_T) |{\cal F}_t \big]  ,
\end{equation}
where $h$ is a smooth function. Weaker assumptions are, in fact, possible for $h$,
as we only need to control the function $Q^{(0)}_t(x)$ defined below rather than $h$.

We introduce the operator
\begin{equation}
{\cal L}_{\rm BS} (\sigma) = \partial_t +\frac{1}{2} \sigma^2 x^2 \partial_x^2 ,
\end{equation}
that is, the standard Black--Scholes operator at zero interest rate and
(constant) volatility $\sigma$.   

We next exploit the fact that the price process is a
martingale to obtain an approximation, by constructing an explicit
function  $Q_t^\eps(x)$, so that $Q_T^\eps(x)=h(x)$ and  
$Q_t^\eps(X_t)$ is a martingale up to first-order terms. Then  $Q_t^\eps(X_t)$ 
gives the approximation  for $M_t$ to this order.
   
The following proposition gives the first-order correction to the 
expression for  the martingale $M_t$ in the regime of small $\eps$.
\begin{proposition}
\label{prop:main}
When $\eps$ is small,
we have
\begin{equation}
M_t = Q_t^\eps(X_t) +o(\eps^{1-H}) ,
\end{equation}
where
\begin{equation}
\label{def:Qt}
Q_t^\eps(x) = Q_t^{(0)}(x) +    \big( x^2\partial_x^2 Q_t^{(0)}(x)\big) \phi_t^\eps
+\eps^{1-H} \widetilde{\sigma} \rho Q_t^{(1)}(x)   .
\end{equation}
The function $Q_t^{(0)}(x) $ is deterministic and given by the Black--Scholes formula with constant volatility $\overline{\sigma}$,
\begin{equation}
\label{eq:bs0}
{\cal L}_{\rm BS} (\overline\sigma) Q_t^{(0)}(x) =0,   \quad \quad Q_T^{(0)}(x) = h(x) .
\end{equation}
The parameters $\overline{\sigma}^2$ and $\widetilde{\sigma}$ are deterministic and given by
\begin{equation} 
\overline{\sigma}^2 = \left< F^2\right> = \int_\RR F(\sigma_{{\rm ou}} z )^2 p(z) dz,
\quad \quad 
\widetilde{\sigma} =\left< F\right> = \int_\RR F(\sigma_{{\rm ou}} z ) p(z) dz,
\end{equation}
where $p(z)$ is the pdf of the standard normal distribution. 
The random component $\phi_t^\eps$ is given by
\begin{equation}
\label{def:phit}
\phi_t^\eps= \EE\Big[ \frac{1}{2} \int_t^T \big( (\sigma_s^\eps)^2 -\overline{\sigma}^2 \big) ds \big| {\cal F}_t\Big].
\end{equation}
The function $Q_t^{(1)}(x) $ is the deterministic correction
\begin{equation}
\label{def:Q1t}
Q_t^{(1)}(x) = x \partial_x \big(x^2 \partial_x^2 Q_t^{(0)}(x)  \big) D_{t} ,
\end{equation}
with $D_{t}$ defined by
\begin{equation}
D_{t} = \overline{D} (T-t)^{H+\frac{1}{2} }, \quad \quad 
\overline{D} =  \frac{
\left< FF' \right>}{\Gamma(H+\frac{3}{2})}
= \frac{
1
}{\Gamma(H+\frac{3}{2})}
\int_\RR FF'(\sigma_{{\rm ou}} z ) p(z) d z
 .
\label{def:DtT}
\end{equation}
\end{proposition}
As shown in Lemma \ref{lem:3} (first item),
as $\eps \to 0$, the zero-mean random variable $\eps^{H-1} \phi_t^\eps$
has a variance that converges to $ \sigma_{\phi}^2 (T-t)^{2H} $, with
\begin{equation}
\label{def:sigmaphi}
\sigma_{\phi}^2 = 
\left<FF'\right>^2\Big( \frac{1}{\Gamma(2H+1) \sin (\pi H)} -\frac{1}{2H \Gamma(H+\frac{1}{2})^2}\Big) .
\end{equation}
Moreover,  it converges in distribution to a Gaussian random variable  with mean zero and variance $ \sigma_{\phi}^2 (T-t)^{2H} $.
This shows that the two corrective terms in (\ref{def:Qt}) are of the same order $\eps^{1-H}$,
but the first one is random, zero-mean and approximately Gaussian distributed,
while the second one is deterministic.

\begin{proof}
For any smooth function $q_t(x)$, we have by It\^o's formula
\begin{eqnarray*}
dq_t(X_t) &=& \partial_t q_t(X_t) dt +  \big( x \partial_x q_t\big) (X_t) \sigma_t^\eps dW_t^*
+\frac{1}{2}  \big( x^2 \partial_x^2 q_t\big) (X_t) (\sigma_t^\eps)^2 dt\\
&=&
{\cal L}_{\rm BS}(\sigma_t^\eps) q_t(X_t) dt +  \big( x \partial_x q_t\big) (X_t) \sigma_t^\eps dW_t^*  ,
\end{eqnarray*}
where the last term is a martingale.
Therefore, by (\ref{eq:bs0}), we have
\begin{eqnarray}
\label{eq:Q01}
dQ_t^{(0)}(X_t)  &=&
\frac{1}{2}\big( (\sigma_t^\eps)^2-\overline{\sigma}^2\big) \big( x^2 \partial_x^2 \big) Q_t^{(0)}(X_t) dt + dN^{(0)}_t ,
\end{eqnarray}
where $N_t^{(0)}$ is a martingale
$$
dN_t^{(0)}= \big( x \partial_x \big) Q_t^{(0)} (X_t) \sigma_t^\eps dW_t^* .
$$
Note also that in Eq.~(\ref{eq:Q01})  (and below), we use the notation
\begin{eqnarray*}
\big( x^2 \partial_x^2 \big) Q_t^{(0)}(X_t) = \left(  \big( x^2 \partial_x^2 \big) Q_t^{(0)}(x) \right)\big|_{ x=X_t}  . 
\end{eqnarray*} 

Let $\phi_t^\eps$ be defined by (\ref{def:phit}).
We have
$$
\phi_t^\eps = \psi_t^\eps - \frac{1}{2} \int_0^t \big( (\sigma_s^\eps)^2 -\overline{\sigma}^2 \big) ds ,
$$
where the martingale $\psi_t^\eps$ is defined by
\begin{equation}
\label{def:Kt}
\psi_t^\eps = 
\EE \Big[  \frac{1}{2} \int_0^T \big( (\sigma_s^\eps)^2 -\overline{\sigma}^2 \big) ds \big| {\cal F}_t\Big] .
\end{equation}
We can write
$$
\frac{1}{2}\big( (\sigma_t^\eps)^2-\overline{\sigma}^2\big)
 \big( x^2 \partial_x^2 \big) Q_t^{(0)}(X_t) dt =
 \big( x^2 \partial_x^2 \big) Q_t^{(0)}(X_t) d\psi_t^\eps -
 \big( x^2 \partial_x^2 \big) Q_t^{(0)}(X_t) d\phi_t^\eps .
$$
By It\^o's formula,
\begin{eqnarray*}
d \big[ \phi_t^\eps \big( x^2 \partial_x^2 \big) Q_t^{(0)}(X_t)\big] &=&
\big( x^2 \partial_x^2 \big) Q_t^{(0)}(X_t) d\phi_t^\eps+
\big( x\partial_x\big(  x^2 \partial_x^2 \big)\big) Q_t^{(0)}(X_t) \sigma_t^\eps \phi_t^\eps dW_t^*
\\
&&+ {\cal L}_{\rm BS}(\sigma_t^\eps)  \big(  x^2 \partial_x^2 \big) Q_t^{(0)}(X_t)  \phi_t^\eps dt
\\
&&
+ \big( x\partial_x\big(  x^2 \partial_x^2 \big)\big) Q_t^{(0)}(X_t) \sigma_t^\eps d\left< \phi^\eps ,W^*\right>_t  .
\end{eqnarray*}
Because ${\cal L}_{\rm BS} (\sigma_t^\eps)= {\cal L}_{\rm BS}(\overline{\sigma}) + \frac{1}{2} \big( (\sigma_t^\eps)^2 - \overline{\sigma}^2\big)\big( x^2 \partial_x^2 \big)$
and ${\cal L}_{\rm BS}(\overline{\sigma})  \big(  x^2 \partial_x^2 \big) Q_t^{(0)}(x)=0$, this gives
\begin{eqnarray*}
d \big[ \phi_t^\eps \big( x^2 \partial_x^2 \big) Q_t^{(0)}(X_t)\big] &=&
-\frac{1}{2} \big( (\sigma_t^\eps)^2 - \overline{\sigma}^2\big)
\big( x^2 \partial_x^2 \big) Q_t^{(0)}(X_t) dt\\
&&+\frac{1}{2} \big( (\sigma_t^\eps)^2 - \overline{\sigma}^2\big)
\big( x^2\partial_x^2\big(  x^2 \partial_x^2 \big)\big)  Q_t^{(0)}(X_t) \phi_t^\eps dt\\
&&+ \big( x\partial_x\big(  x^2 \partial_x^2 \big) \big) Q_t^{(0)}(X_t) \sigma_t^\eps  d\left< \phi^\eps ,W^*\right>_t  
\\
&&  +
\big( x\partial_x\big(  x^2 \partial_x^2 \big)\big) Q_t^{(0)}(X_t) \sigma_t^\eps \phi_t^\eps dW_t^*+
\big( x^2 \partial_x^2 \big) Q_t^{(0)}(X_t) d\psi_t^\eps
.
\end{eqnarray*}
We have $\left< \phi^\eps ,W^*\right>_t = \left< \psi^\eps ,W^*\right>_t = \rho \left< \psi^\eps ,W\right>_t$, and therefore
\begin{eqnarray*}
d \big[( \phi_t^\eps \big( x^2 \partial_x^2 \big) Q_t^{(0)}(X_t)\big]
&=& -\frac{1}{2} \big( (\sigma_t^\eps)^2 - \overline{\sigma}^2\big)
\big( x^2 \partial_x^2 \big) Q_t^{(0)}(X_t) dt\\
&&+\frac{1}{2} \big( (\sigma_t^\eps)^2 - \overline{\sigma}^2\big)
\big( x^2\partial_x^2\big(  x^2 \partial_x^2 \big)\big)   Q_t^{(0)}(X_t) \phi_t^\eps dt
\\
&&+ \rho \big( x\partial_x\big(  x^2 \partial_x^2 \big) \big) Q_t^{(0)}(X_t)  \sigma_t^\eps    d\left< \psi^\eps ,W \right>_t   \\
&&+ dN^{(1)}_t ,
\end{eqnarray*}
where $N^{(1)}_t$ is a martingale
$$
dN^{(1)}_t =\big( x\partial_x\big(  x^2 \partial_x^2 \big)\big) Q_t^{(0)}(X_t) \sigma_t^\eps \phi_t^\eps dW_t^*
+
 \big( x^2 \partial_x^2 \big) Q_t^{(0)}(X_t) d\psi_t^\eps .
$$
Therefore,
\begin{eqnarray}
\nonumber
d \big[ Q_t^{(0)}(X_t)+ \phi_t^\eps \big( x^2 \partial_x^2 \big) Q_t^{(0)}(X_t)\big]
&=&  \frac{1}{2} \big( (\sigma_t^\eps)^2 - \overline{\sigma}^2\big)
\big( x^2\partial_x^2\big(  x^2 \partial_x^2 \big)\big) Q_t^{(0)}(X_t) \phi_t^\eps dt
\\
\nonumber
&&+ \rho   \big( x\partial_x\big(  x^2 \partial_x^2 \big) \big) Q_t^{(0)}(X_t) \sigma_t^\eps  d\left< \psi^\eps ,W \right>_t   \\
&&+  dN^{(0)}_t +dN^{(1)}_t  .
\end{eqnarray}

The deterministic function $Q^{(1)}_t$ defined by (\ref{def:Q1t}) satisfies
$$
{\cal L}_{\rm BS}(\overline{\sigma}) Q^{(1)}_t(x) = -
\big( x\partial_x \big( x^2 \partial_x^2 Q^{(0)}_t(x)\big)\big) \theta_{t} , \quad \quad Q^{(1)}_T(x) = 0,
$$
where $\theta_{t} = -dD_t/dt$ is such that
$$
d\left< \psi^\eps,W\right>_t = \big( \eps^{1-H}   \theta_{t} + \widetilde{\theta}^\eps_{t} \big)   dt,
$$
as shown in Lemmas \ref{lem:1}-\ref{lem:2} with $\widetilde{\theta}^\eps_{t}$ characterized 
in Eq. (\ref{eq:phit}).
By applying It\^o's formula, we obtain
\begin{eqnarray*}
dQ_t^{(1)}(X_t)  
&=&
{\cal L}_{\rm BS}(\sigma_t^\eps) Q_t^{(1)}(X_t) dt +  \big( x \partial_x \big) Q_t^{(1)} (X_t) \sigma_t^\eps dW_t^* \\
&=& {\cal L}_{\rm BS}(\overline\sigma) Q_t^{(1)}(X_t) dt + 
\frac{1}{2} \big( (\sigma_t^\eps)^2 -\overline{\sigma}^2 \big) \big( x^2 \partial_x^2 \big) Q_t^{(1)}(X_t) dt  \\
&&
+
 \big( x \partial_x \big) Q_t^{(1)}  (X_t) \sigma_t^\eps dW_t^*  \\
 &=& \frac{1}{2} \big( (\sigma_t^\eps)^2 -\overline{\sigma}^2 \big) \big( x^2 \partial_x^2 \big) Q_t^{(1)}(X_t) dt 
 - \big( x\partial_x \big( x^2 \partial_x^2 \big)\big)Q^{(0)}_t(X_t) \theta_{t} dt
+ dN^{(2)}_t ,
\end{eqnarray*}
where $N^{(2)}_t$ is a martingale
$$
dN^{(2)}_t =  \big( x \partial_x\big) Q_t^{(1)} (X_t) \sigma_t^\eps dW_t^*   .
$$
Therefore,
\begin{eqnarray}
\nonumber
&&
d \big[ Q_t^{(0)}(X_t)+ \phi_t^\eps \big( x^2 \partial_x^2 \big) Q_t^{(0)}(X_t)
+\eps^{1-H} \rho \widetilde{\sigma} Q_t^{(1)}(X_t)  
\big] \\
\nonumber
&&
 = \frac{1}{2} \big( (\sigma_t^\eps)^2 - \overline{\sigma}^2\big)
\big( x^2\partial_x^2\big(  x^2 \partial_x^2 \big)\big) Q_t^{(0)}(X_t) \phi_t^\eps dt
+\frac{\eps^{1-H}}{2} \rho \widetilde{\sigma} \big( (\sigma_t^\eps)^2 -\overline{\sigma}^2 \big) \big( x^2 \partial_x^2 \big) Q_t^{(1)}(X_t) dt  \\
\nonumber
 &&
\quad 
+\eps^{1-H}  \rho   \big( x\partial_x\big(  x^2 \partial_x^2 \big) \big) Q_t^{(0)}(X_t) 
(\sigma_t^\eps-\widetilde{\sigma}) \theta_{t}dt 
+  \rho   \big( x\partial_x\big(  x^2 \partial_x^2 \big) \big) Q_t^{(0)}(X_t)  \sigma_t^\eps \widetilde{\theta}_{t}^{\eps}dt \\
&&
\quad
+ d N^{(0)}_t + dN^{(1)}_t +\eps^{1-H} \rho \widetilde{\sigma}dN^{(2)}_t .
\label{eq:proof1b}
\end{eqnarray}
We next  show that the first four terms of the right-hand side are smaller than $\eps^{1-H}$.
We introduce, for any $t \in [0,T]$,
\begin{eqnarray}
R^{(1)}_{t,T} &=& \int_t^T  \frac{1}{2} \big( x^2\partial_x^2\big(  x^2 \partial_x^2 \big)\big) Q_s^{(0)}(X_s)  \big( (\sigma_s^\eps)^2 - \overline{\sigma}^2\big) \phi_s^\eps
 ds , \\
R^{(2)}_{t,T} &=& \int_t^T \frac{\eps^{1-H}}{2} \rho \widetilde{\sigma} \big( x^2 \partial_x^2 \big) Q_s^{(1)}(X_s)  \big( (\sigma_s^\eps)^2 -\overline{\sigma}^2 \big)ds , \\
R^{(3)}_{t,T} &=& \int_t^T \eps^{1-H}  \rho   \big( x\partial_x\big(  x^2 \partial_x^2 \big) \big) Q_s^{(0)}(X_s) 
 \theta_{s}   (\sigma_s^\eps-\widetilde{\sigma})ds  , \\
R^{(4)}_{t,T} &=& \int_t^T \rho   \big( x\partial_x\big(  x^2 \partial_x^2 \big) \big) Q_s^{(0)}(X_s)  \sigma_s^\eps \widetilde{\theta}_{s}^{\eps}ds .
\end{eqnarray}

We show that,  for $j=1,2,3,4$,
\begin{equation}
\label{eq:estimeRj}
\displaystyle \lim_{\eps \to 0} \eps^{H-1} \sup_{t \in [0,T]} \EE \big[ (R^{(j)}_{t,T})^2 \big]^{1/2} =0 .
\end{equation}

{\it Step 1: Proof of (\ref{eq:estimeRj}) for $j=1$.}\\
We denote
$$
Y^{(1)}_s = \big( x^2\partial_x^2\big(  x^2 \partial_x^2 \big)\big) Q_s^{(0)}(X_s)
$$
and
\begin{equation}
\label{def:gammaeps}
\gamma_t^\eps = \frac{1}{2} \int_0^t \big( (\sigma_s^\eps)^2 -\overline{\sigma}^2 \big) \phi_s^\eps ds ,
\end{equation}
so that we can write 
$$
R^{(1)}_{t,T} = \int_t^T Y^{(1)}_s \frac{d\gamma_s^\eps }{ds} ds .
$$
Note that $Y^{(1)}_s$ is a bounded semimartingale with bounded quadratic variations,
so that its mean square increments $\EE[ (Y^{(1)}_s-Y^{(1)}_{s'})^2]$ are uniformly bounded by $K|s-s'|$.
Let $N$ be a positive integer. We denote $t_k=t+(T-t)k/N$. We have
\begin{eqnarray*}
R^{(1)}_{t,T} &=& \sum_{k=0}^{N-1} \int_{t_k}^{t_{k+1}}Y^{(1)}_s \frac{d\gamma_s^\eps }{ds} ds =R^{(1,a)}_{t,T} +R^{(1,b)}_{t,T}  ,\\
R^{(1,a)}_{t,T} &=&\sum_{k=0}^{N-1} \int_{t_k}^{t_{k+1}}Y^{(1)}_{t_k} \frac{d\gamma_s^\eps }{ds} ds = \sum_{k=0}^{N-1} Y^{(1)}_{t_k} \big( \gamma_{t_{k+1}}^\eps -
 \gamma_{t_{k}}^\eps \big) , \\
 R^{(1,b)}_{t,T} &=&\sum_{k=0}^{N-1} \int_{t_k}^{t_{k+1}}\big( Y^{(1)}_{s}-Y^{(1)}_{t_k}\big)  \frac{d\gamma_s^\eps }{ds} ds  .
\end{eqnarray*}
Note that we obtain by Minkowski's inequality 
\begin{eqnarray*}
\EE \big[ (R^{(1,a)}_{t,T})^2 \big]^{1/2} &\leq & 2 \sum_{k=0}^{N} \|Y^{(1)}\|_\infty 
\EE[ (\gamma_{t_{k}}^\eps)^2]^{1/2} \leq  2  (N+1) \|Y^{(1)}\|_\infty \sup_{s\in [0,T]}
\EE[ (\gamma_{s}^\eps)^2]^{1/2}  ,
\end{eqnarray*}
 so that, by Lemma \ref{lem:4}, we have, for any fixed $N$,
$$
 \lim_{\eps \to 0} \eps^{H-1} \sup_{t \in [0,T]} \EE \big[ (R^{(1,a)}_{t,T})^2 \big]^{1/2} =0 .
$$
On the other hand,
\begin{eqnarray*}
\EE \big[ (R^{(1,b)}_{t,T})^2 \big]^{1/2} &\leq & \|F\|_\infty^2 \sum_{k=0}^{N-1} \int_{t_k}^{t_{k+1}} \EE[ \big( Y^{(1)}_{s}-Y^{(1)}_{t_k}\big)^4]^{1/4}
\EE [ (\phi_s^\eps )^4]^{1/4} ds \\
&\leq & K \sum_{k=0}^{N-1} \int_{t_k}^{t_{k+1}} (s-t_k)^{1/2} ds
 \sup_{s\in [0,T]} \EE [ (\phi_s^\eps )^4]^{1/4}   \\
&\leq & \frac{K'}{\sqrt{N}} \sup_{s\in [0,T]} \EE [ (\phi_s^\eps )^4]^{1/4}   .
\end{eqnarray*}
 Therefore, 
by Lemma \ref{lem:3} (fourth item), we get
$$
 \limsup_{\eps \to 0} \eps^{H-1} \sup_{t \in [0,T]} \EE \big[ (R^{(1)}_{t,T})^2 \big]^{1/2} \leq  \limsup_{\eps \to 0} \eps^{H-1} \sup_{t \in [0,T]} \EE \big[ (R^{(1,b)}_{t,T})^2 \big]^{1/2} \leq \frac{K'}{\sqrt{N}} .
$$
Because this is true for any $N$, we get the desired result.\\

{\it Step 2: Proof of (\ref{eq:estimeRj}) for $j=2$.}\\
We denote
$$
Y^{(2)}_s =   \rho \widetilde{\sigma} \big( x^2 \partial_x^2 \big) Q_s^{(1)}(X_s)  
$$
and
\begin{equation}
\label{def:kappaeps}
\kappa_t^\eps = \frac{\eps^{1-H}}{2}  \int_0^t \big( (\sigma_s^\eps)^2 -\overline{\sigma}^2 \big)ds  ,
\end{equation}
so that we can write 
$$
R^{(2)}_{t,T} = \int_t^T Y^{(2)}_s \frac{d\kappa_s^\eps }{ds} ds .
$$
Note that $Y^{(2)}_s$ is a bounded semimartingale with bounded quadratic variations.
Let $N$ be a positive integer. We denote as above $t_k=t+(T-t)k/N$. We then have
\begin{eqnarray*}
R^{(2)}_{t,T} &=& \sum_{k=0}^{N-1} \int_{t_k}^{t_{k+1}}Y^{(2)}_s \frac{d\kappa_s^\eps }{ds} ds =R^{(2,a)}_{t,T} +R^{(2,b)}_{t,T}  ,\\
R^{(2,a)}_{t,T} &=&\sum_{k=0}^{N-1} \int_{t_k}^{t_{k+1}}Y^{(2)}_{t_k} \frac{d\kappa_s^\eps }{ds} ds = \sum_{k=0}^{N-1} Y^{(2)}_{t_k} \big( \kappa_{t_{k+1}}^\eps -
 \kappa_{t_{k}}^\eps \big) , \\
 R^{(2,b)}_{t,T} &=&\sum_{k=0}^{N-1} \int_{t_k}^{t_{k+1}}\big( Y^{(2)}_{s}-Y^{(2)}_{t_k}\big)  \frac{d\kappa_s^\eps }{ds} ds  .
\end{eqnarray*}
Then,  on the one hand,
\begin{eqnarray*}
\EE \big[ (R^{(2,a)}_{t,T})^2 \big]^{1/2} &\leq & 2 \sum_{k=0}^{N} \|Y^{(2)}\|_\infty 
\EE[ (\kappa_{t_{k}}^\eps)^2]^{1/2}  \leq  2  (N+1) \|Y^{(2)}\|_\infty \sup_{s\in [0,T]}
\EE[ (\kappa_{s}^\eps)^2]^{1/2}  ,
\end{eqnarray*}
so that, by Lemma \ref{lem:6}, we obtain
$$
 \lim_{\eps \to 0} \eps^{H-1} \sup_{t \in [0,T]} \EE \big[ (R^{(2,a)}_{t,T})^2 \big]^{1/2} =0 .
$$
On the other hand,
\begin{eqnarray*}
\EE \big[ (R^{(2,b)}_{t,T})^2 \big]^{1/2} &\leq & \eps^{1-H} \|F\|_\infty^2 \sum_{k=0}^{N-1} \int_{t_k}^{t_{k+1}} \EE[ \big( Y^{(2)}_{s}-Y^{(2)}_{t_k}\big)^2]^{1/2}ds \\
&\leq & K \eps^{1-H} \sum_{k=0}^{N-1} \int_{t_k}^{t_{k+1}} (s-t_k)^{1/2} ds  \\
&\leq & \frac{K' \eps^{1-H}}{\sqrt{N}}   .
\end{eqnarray*}
Therefore, we get
$$
 \limsup_{\eps \to 0} \eps^{H-1} \sup_{t \in [0,T]} \EE \big[ (R^{(2)}_{t,T})^2 \big]^{1/2} \leq  \limsup_{\eps \to 0} \eps^{H-1} \sup_{t \in [0,T]} \EE \big[ (R^{(2,b)}_{t,T})^2 \big]^{1/2} \leq \frac{K'}{\sqrt{N}} .
$$
Because this is true for any $N$, we get the desired result.\\

{\it Step 3: Proof of (\ref{eq:estimeRj}) for $j=3$.}\\
This proof follows the same lines as the proof of Step 2 with
\begin{equation}
\label{def:etaeps}
\eta_t^\eps = \eps^{1-H}  \int_0^t \big( \sigma_s^\eps  -\widetilde{\sigma} \big)   ds   ,
\end{equation}
instead of $\kappa_{t}^\eps$, and using the fact that $\theta_t$ is bounded. We then  get the desired result by Lemma \ref{lem:5}.\\

{\it Step 4: Proof of (\ref{eq:estimeRj}) for $j=4$.}\\
We have
$$
\EE \big[ (R^{(4)}_{t,T})^2 \big]^{1/2}  \leq K \int_t^T  \EE \big[  (\widetilde{\theta}_{s}^{\eps})^2 \big]^{1/2} ds  \leq K' \sup_{s\in [0,T]}  \EE \big[  (\widetilde{\theta}_{s}^{\eps})^2 \big]^{1/2} .
$$
By Lemma \ref{lem:2}, we obtain
$$
 \lim_{\eps \to 0} \eps^{H-1} \sup_{t \in [0,T]} \EE \big[ (R^{(4)}_{t,T})^2 \big]^{1/2} =0 .
$$

We can now complete the proof of Proposition \ref{prop:main}.
In (\ref{def:Qt}),  we introduced the approximation 
$$
Q_t^\eps(x) = 
Q_t^{(0)}(x)+ \phi_t^\eps \big( x^2 \partial_x^2 \big) Q_t^{(0)}(x)
+\eps^{1-H} \rho \widetilde{\sigma} Q_t^{(1)}(x).
$$
We then have
$$
Q_T^\eps(x) = h(x),
$$
because $Q_T^{(0)}(x)=h(x)$, $\phi_T^\eps=0$, and $Q^{(1)}_T(x)=0$.
Let us denote
\begin{eqnarray}
R_{t,T} &=&R^{(1)}_{t,T}+
R^{(2)}_{t,T}+
R^{(3)}_{t,T}+
R^{(4)}_{t,T},\\
N_t &=&  \int_0^t d N^{(0)}_s + dN^{(1)}_s +\eps^{1-H} \rho \widetilde{\sigma}dN^{(2)}_s  .
\end{eqnarray}
By (\ref{eq:proof1b}) we have
$$
Q_T^\eps(X_T)
 - Q_t^\eps(X_t) = R_{t,T} + N_T-  N_t.
$$
Therefore,
\begin{eqnarray}
\nonumber
M_t &=& \EE \big[ h(X_T) |{\cal F}_t \big] = 
\EE \big[ Q_T^\eps(X_T) |{\cal F}_t \big] = Q_t^\eps(X_t) +\EE \big[ R_{t,T} |{\cal F}_t \big]+
\EE \big[ N_T-N_t |{\cal F}_t \big] \\
&=&
Q_t^\eps(X_t) +\EE \big[ R_{t,T} |{\cal F}_t \big]
 ,
\end{eqnarray}
which gives the desired result, because $\EE \big[ R_{t,T}  |{\cal F}_t \big]$ is of order $o(\eps^{1-H})$ in $L^2$.
\end{proof}

\section{Call Price Correction and Implied Volatility}\label{sec:implied}
We denote the Black--Scholes call price, with current time $t$, maturity $T$, strike $K$,
underlying value $x$, and volatility $\sigma$, by $C_{\rm BS}(t,x;K,T;\sigma)$, so that
$Q_t^{(0)}$ in Eq. (\ref{eq:bs0}) is
$$
 Q_t^{(0)}(x)=C_{\rm BS}(t,x;K,T;\overline{\sigma}) .
 $$
Indeed,   $C_{\rm BS}$ gives an explicit formula for the price when
volatility is constant. In the case with stochastic volatility as considered here, no explicit 
pricing formula exists. As shown in Eq.~(\ref{def:Qt}), however,  we can 
get an asymptotic expression for the price in the case with the stochastic volatility model (\ref{def:stochmodel0})
as a correction to $Q_t^{(0)}(x)$, the Black--Scholes price
evaluated at the effective, or  ``homogenized'', volatility $\bar\sigma$.  
 Here, we show that this corrected price takes on a rather simple, generic form in the two 
  parameters: relative time to maturity and moneyness.
  This representation then leads to a simple representation for the implied volatility,
  as we show below. 
The long-range character of the volatility fluctuations indeed has a strong impact 
on the form of the implied volatility, and  this observation is important in a calibration context.

   We denote the time to maturity by $\tau=T-t$, and we introduce the 
 characteristic  diffusion time $\bar\tau = 2/\overline{\sigma}^2$ and
  the dimensionless effective skewness factor
\begin{equation}\label{eq:aF}
 a_F =    \eps^{1-H}  \frac{\rho  \widetilde{\sigma} \overline{D} 
  \bar\tau^{H} }{ 2^{3/2} \overline{\sigma}   } =
 \eps^{1-H} \frac{ \widetilde{\sigma} 
 \rho 
  \left<FF'\right> \bar\tau^{H} }{ 2^{3/2} \overline{\sigma} \Gamma(H+{3}/{2})}  ,
\end{equation}
 with  $\overline{\sigma},$  $\widetilde{\sigma}$   
 and $\overline{D}$ given in  
Proposition \ref{prop:main} and  the correlation $\rho$ introduced in Eq. (\ref{eq:corr}).

\begin{lemma}
\label{lem:1N}
    The price correction in  Eq.  (\ref{def:Qt}),   normalized by the strike $K$, can 
be written in the form
 \begin{eqnarray}    \nonumber 
   &&   
  \frac{1}{K}   \left(\phi_t^\eps \big( x^2 \partial_x^2 \big) Q_t^{(0)}(x)
+\eps^{1-H} \rho \widetilde{\sigma} Q_t^{(1)}(x)   \right)   \\ &&  =   
    \left(  \frac{    e^{-d_1^2/2} \frac{x}{K}   }{  \sqrt{ \pi }  } \right) 
\left\{
\frac{\phi_t^\eps}{{2}} \left( \frac{\tau}{\bar\tau} \right)^{-1/2}  
    +
   a_F   \left[
     \left( \frac{\tau}{\bar\tau}\right)^{H}    +      \left( \frac{\tau}{\bar\tau}\right)^{H-1}    {\log \Big(\frac{K}{x} \Big)}  \right]   
    \right\} ,  \label{eq:pricecorr3}  
 \end{eqnarray}
 with
 \begin{eqnarray}
  d_1 =  \sqrt{ \frac{  \bar{\tau} }{ 2 \tau } } \Big[ \frac{\tau}{\bar\tau}  - \log \Big(\frac{K}{x}\Big)\Big] .
 \end{eqnarray}
 \end{lemma}
Here,  the  dimensionless random and deterministic 
correction coefficients 
 are small,
 \begin{eqnarray}
   \phi_t^\eps  =O\left( \left(  \frac{ \eps }{ \bar\tau } \right)^{1-H} \left( \frac{   \tau }{ \bar\tau } \right)^{H} \right)  
     , \quad \quad a_F = O\left(  \frac{ \eps }{ \bar\tau } \right)^{1-H}  ,
 \end{eqnarray}
 where  we used the fact that  $\phi_t^\eps$ as defined in  Proposition \ref{prop:main} 
 is centered and with standard deviation
 \begin{equation}\label{eq:phisd}
{\rm Var} \big( {\phi_t^\eps}  \big)^{1/2} 
=   
 \left(  \frac{ \eps }{ \bar\tau } \right)^{1-H} \left( \frac{   \tau }{ \bar\tau } \right)^{H} 
 \left(  \bar\tau  {\sigma_\phi } \right)   + o(\eps^{1-H}),
\end{equation}
with  $\sigma_\phi$ defined by Eq.~(\ref{def:sigmaphi}) (see also Eq.~(\ref{eq:stdev}) in Lemma \ref{lem:3}).
We comment in more detail about the 
statistical structure of $\phi_t^\eps$ in the next section. 

It follows from the above that the  normalized price correction  depends on the two parameters - the moneyness  $K/x$ and
the relative time to maturity $\tau/\bar\tau$ -  and exhibits a term structure in  fractional powers of 
relative time to maturity. 

In Figure \ref{fig1_H06} we show the relative price correction 
in Eq.~(\ref{eq:pricecorr3}) as a function of relative time to maturity $\tau/\bar\tau$   for three values of the moneyness $K/x$.
The solid lines plot the mean relative price correction, and the dashed lines give the mean plus/minus
one standard deviation.  We use here $H=0.6$, $a_F=0.1$, and $\left((\eps/\bar\tau)^{(1-H)} \bar\tau \sigma_\phi\right)=0.04$. 
The mean relative price correction is largest for a mid-range of times to maturity.  For very short times to maturity relative to the 
characteristic diffusion time, the effect of the volatility  fluctuations  are small, while for long times the rapid mean reversion
``averages'' out the effect of the fluctuations.  
Note,  however,  that at-the-money  the random component 
of the price correction decays slowly as
\[
   \left(    \frac{\tau}{\bar{\tau}}    \right)^{H-1/2} , 
\]  
as $\tau \to 0$ while  ``around-the-money'' with the moneyness $K/x$ being different from one, the decay has the form
\[
   \left(  {\frac{\tau}{\bar{\tau}}}   \right)^{H-1/2}  \exp \Big( -  \frac{\bar{\tau} |\log(K/x)|^2}{4\tau}  \Big) .
\] 
This reflects the fact that the vega is diverging in this limit so that
the sensitivity to volatility fluctuations becomes strong. We remark that this
would affect calibration schemes using at-the-money data.
Moreover,  results regarding short time asymptotics 
for the coherent implied volatility become questionable in this context as the dominating 
contribution comes from the random component of the price correction.
Note also that the parameters are not calibrated to market data; this will be considered in another publication.

\begin{figure}
\begin{center}
\begin{tabular}{c}
\includegraphics[width=8.4cm]{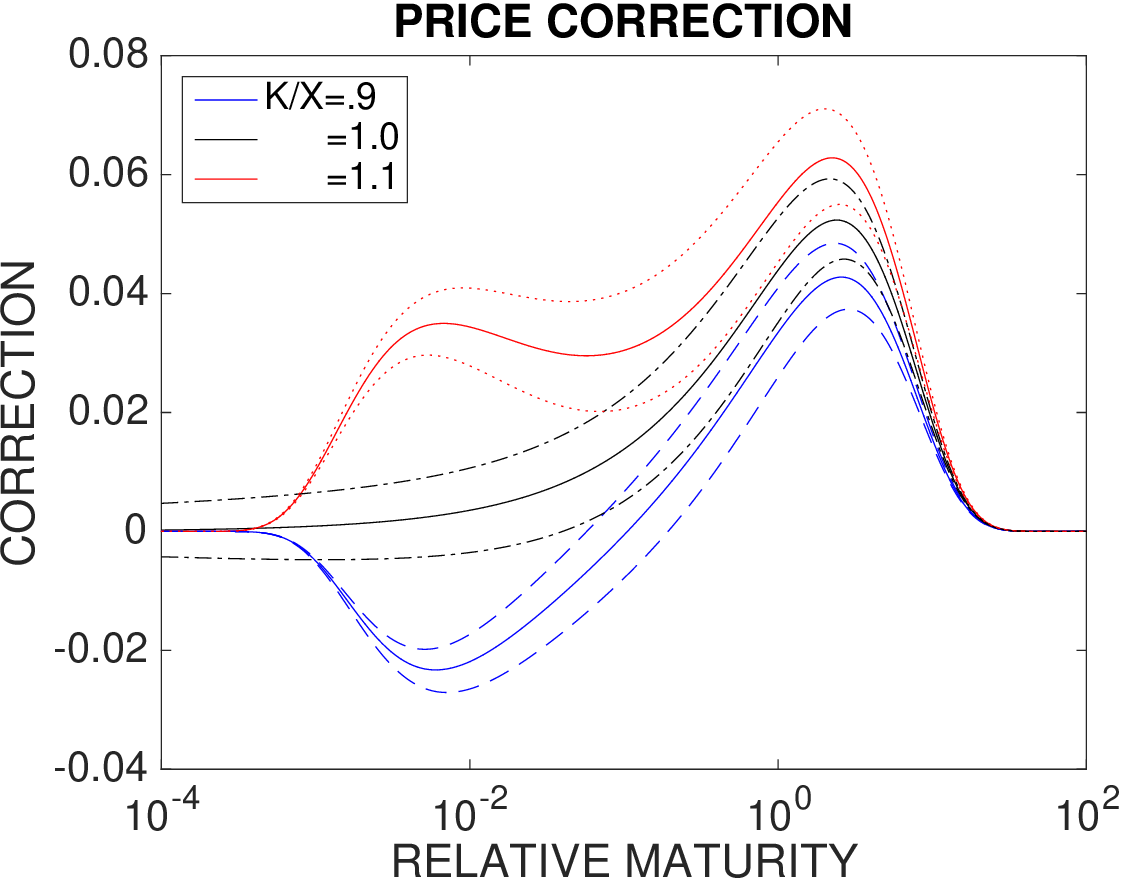}
\end{tabular}
\end{center}
\caption{Price correction as a function of the relative time to maturity $\tau/\bar\tau$.  The  three solid lines correspond (from bottom  to top)
to the mean price correction for $K/X=0.9$, $1.0$, and $1.1$, respectively.   The dashed/dotted  lines correspond to the mean $\pm 1$ standard deviation. 
Here $H=0.6, a_F=0.1$, and $\left((\eps/\bar\tau)^{(1-H)} \bar\tau \sigma_\phi\right)=0.04$. 
\label{fig1_H06}
}
\end{figure}

In Figure \ref{fig_surff}, we show the price  correction surface as a function
of the relative time to maturity  $\tau/\bar\tau$   and the moneyness $K/x$. 
The figure shows that the 
price correction is  large when the time to maturity is of the order of the characteristic diffusion time.

\begin{figure}
\begin{center}
\begin{tabular}{c}
\includegraphics[width=8.4cm]{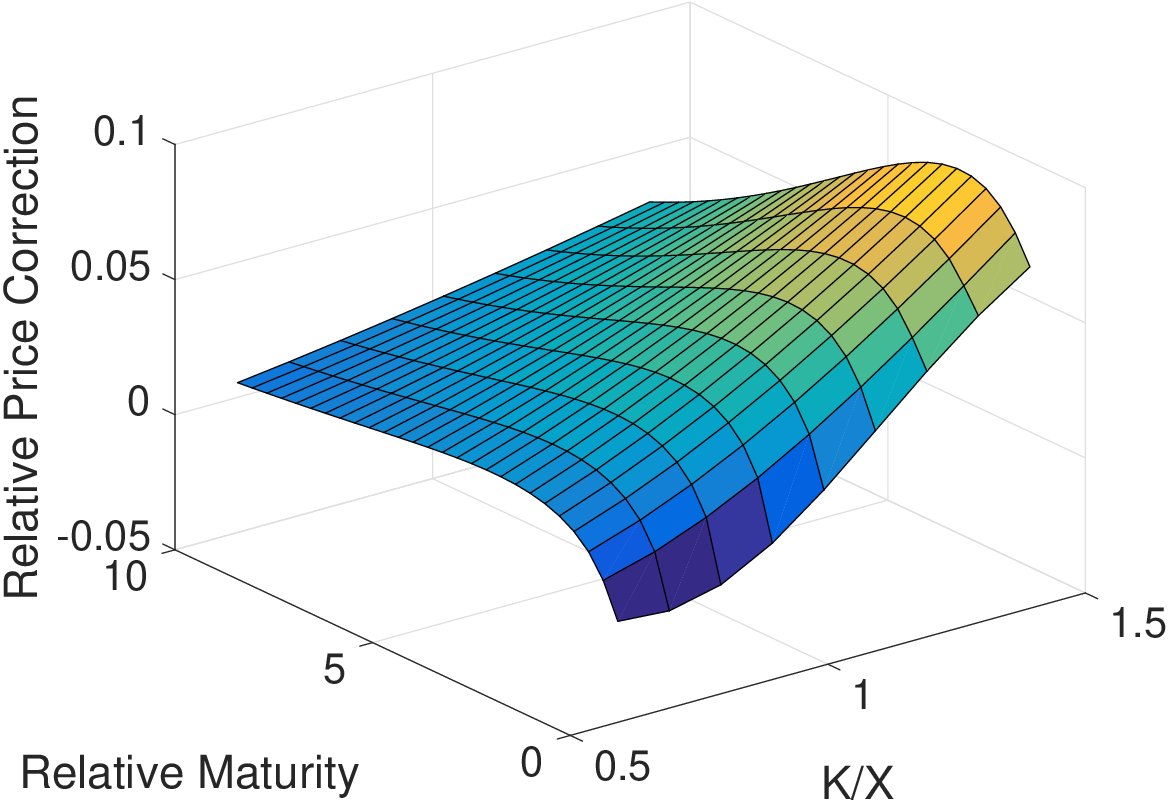}
\end{tabular}
\end{center}
\caption{The price  correction surface as a function of the relative time to maturity $\tau/\bar\tau$ and the moneyness $K/X$.
The parameters are like those in Figure \ref{fig1_H06}.
\label{fig_surff}
}
\end{figure}
 
We next present the proof of Lemma \ref{lem:1N}. 
 \begin{proof}
 \label{lem:price}
  For the  European call option with payoff $h(x)=(x-K)_+$, we have 
\begin{eqnarray*}
\nonumber
C_{\rm BS}(t,x;K,T;\sigma)  &=& 
x \Phi
\left(\frac{1}{ {\sigma}\sqrt{T-t}}  \log \left( \frac{x}{K} \right)+ \frac{ {\sigma} \sqrt{T-t} }{2}  \right)
\\
&&- 
K \Phi
\left(\frac{1}{ {\sigma}\sqrt{T-t}}   \log \left(\frac{x}{K}\right)- \frac{ {\sigma} \sqrt{T-t} }{2} \right),
\end{eqnarray*}
where $\Phi$ is the cumulative distribution function  of the standard normal distribution.
We then have, in particular, the ``Greek'' relationships for the call price
\begin{eqnarray*}
 {\partial_\sigma} C_{\rm BS}  = (T-t) \overline\sigma x^2 \partial_x^2  C_{\rm BS}, \quad \quad
 x\partial_x {\partial_{\sigma}} C_{\rm BS}  =  
 \left(
     \frac{1}{2}   + \frac{\log \frac{K}{x}}{\overline\sigma^2(T-t)}  \right) \partial_\sigma C_{\rm BS}  .
\end{eqnarray*}
We then get
\begin{eqnarray}\label{eq:modes}
 x^2 \partial_x^2  Q_t^{(0)}(x)
    &=&          \frac{1}{\overline\sigma (T-t) }  {\partial_{\bar\sigma}} C_{\rm BS}(t,x;K,T;\overline\sigma) ,    \\
   x \partial_x    x^2 \partial_x^2  Q_t^{(0)}(x)  & = & 
   \left[
     \frac{1}{2\overline\sigma(T-t)}   + \frac{\log \frac{K}{x} }{\overline\sigma^3(T-t)^2}  \right]  {\partial_{\bar\sigma}}
      C_{\rm BS}(t,x;K,T;\overline\sigma) ,
 \end{eqnarray}
 where the ``vega''  is given by
 \begin{eqnarray}\label{eq:vega}
    {\partial_{\sigma}} C_{\rm BS}(t,x;K,T;\overline\sigma)  =
           \frac{x e^{-d_1^2/2} \sqrt{T-t} }{\sqrt{2\pi}}, \quad  \quad
            d_1 = 
             \frac{  \frac{1}{2} \sigma^2 (T-t)  - \log \frac{K}{x} }{ \sigma \sqrt{T-t} }   .
 \end{eqnarray}
 Then, with
 $Q_t^{(1)}(x)$ given  in Eq.  (\ref{def:Q1t}),
 we can identify the form of  the  price correction  as
 \begin{eqnarray}
 \nonumber   
   &&   \hspace*{-0.15in}
     \phi_t^\eps \big( x^2 \partial_x^2 \big) Q_t^{(0)}(x)
+\eps^{1-H} \rho \widetilde{\sigma} Q_t^{(1)}(x)    \\ 
&&     \hspace*{-0.15in}=   \nonumber 
    \phi_t^\eps \big( x^2 \partial_x^2 \big) Q_t^{(0)}(x)
+\eps^{1-H} \rho \widetilde{\sigma}  D(t) x \partial_x    x^2 \partial_x^2  Q_t^{(0)}(x)    \\ 
&&   \hspace*{-0.15in} =
\phi_t^\eps     
\left( \frac{x e^{-d_1^2/2} }{ \overline{\sigma}  \sqrt{2\pi(T-t)}} \right)
+\eps^{1-H}  \left(  \frac{x \rho \widetilde{\sigma}   \overline{D}   e^{-d_1^2/2}  }{  \sqrt{2\pi }  } \right) 
   \left[
     \frac{ ({T-t})^{H}  }{2\overline\sigma }   + \frac{\log \frac{K}{x} }{\overline\sigma^3(T-t)^{1-H}}  \right]   ,
     \label{pc} 
 \end{eqnarray}
 which in turn gives (\ref{eq:pricecorr3}).
\end{proof}

 We next consider the implied volatility associated with the price correction.  
For the stochastic volatility model in Eq. (\ref{def:stochmodel0}),
we  want to identify the implied volatility  $I_t$ so that in terms of the corrected price
in Lemma \ref{prop:main}, we have 
\begin{eqnarray}\label{eq:cal}
 C_{\rm BS}(t,x;K,T;I_t)  =
Q_t^{(0)}(x )+ \phi_t^\eps \big( x^2 \partial_x^2 \big) Q_t^{(0)}(x)
+\eps^{1-H} \rho \widetilde{\sigma} Q_t^{(1)}(x).
\end{eqnarray}

We define the relative implied volatility correction $\delta I_t$ by 
\begin{equation}\label{def:di}
  I_t=\overline{\sigma}(1  + \delta I_t) .
\end{equation}

\begin{lemma}
\label{lem:2N}
The relative implied volatility correction has the form 
 \begin{eqnarray}  \label{eq:implied2} 
\delta I_t  & =&   
 \frac{\phi_t^\eps}{2}   \left( \frac{\tau}{\bar\tau} \right)^{-1} 
 +a_F 
\Big[  \left( \frac{\tau}{\bar\tau} \right)^{H-1/2} +  \left( \frac{\tau}{\bar\tau} \right)^{H-3/2} {\log\Big(\frac{K}{X_t}\Big)} \Big]
    +    o(\eps^{1-H}) ,
\end{eqnarray}
where  $\phi_t^\eps$ is defined  by (\ref{def:phit}) and $a_F$ by (\ref{eq:aF}). 
\end{lemma}

In Figure \ref{fig2_H06}, we show the implied volatility correction 
in Eq.~(\ref{eq:implied2}) as a function of relative time to maturity $\tau/\bar\tau$   for three values of the moneyness $K/x$.
We  again used $H=0.6, a_F=0.1$ and $\left((\eps/\bar\tau)^{(1-H)} \bar\tau \sigma_\phi\right)=0.04$. 
Note that due to the form of the  ``vega''  (which is the sensitivity of the price to the volatility), 
 the form of the implied volatility surface is very different
from that of the price correction.   
\begin{figure}
\begin{center}
\begin{tabular}{c}
\includegraphics[width=8.4cm]{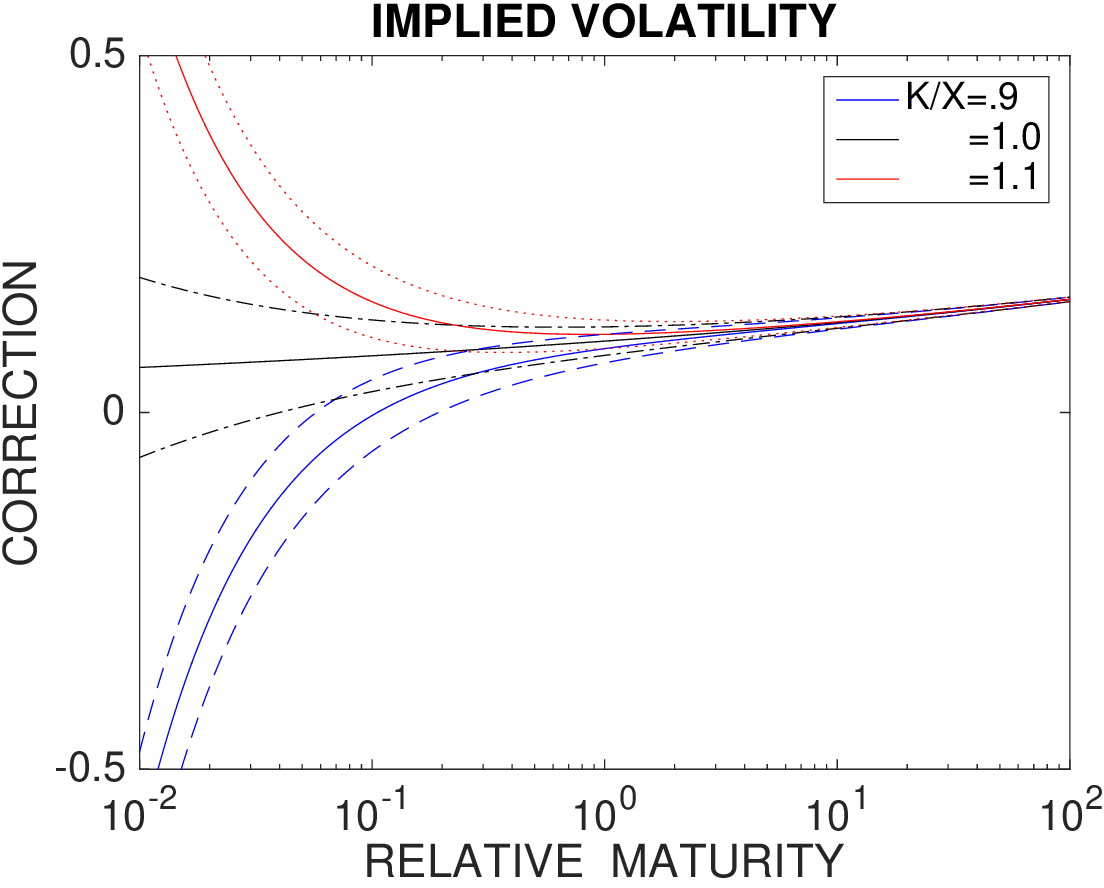}
\end{tabular}
\end{center}
\caption{The implied volatility correction as a function of the relative time to maturity  $\tau/\bar\tau$.   
The  three solid lines correspond (from bottom  to top) to the mean  implied volatility correction
for $K/X=0.9$, $1.0$, and $1.1$, respectively.   The dashed/dotted  lines correspond to the mean  
$\pm 1$ standard deviation. 
\label{fig2_H06}
}
\end{figure}
In Figure \ref{fig_surffi}, we show the implied volatility correction surface as a function
of the relative time to maturity  $\tau/\bar\tau$   and the moneyness $K/x$. 
  \begin{figure}
\begin{center}
\begin{tabular}{c}
\includegraphics[width=8.4cm]{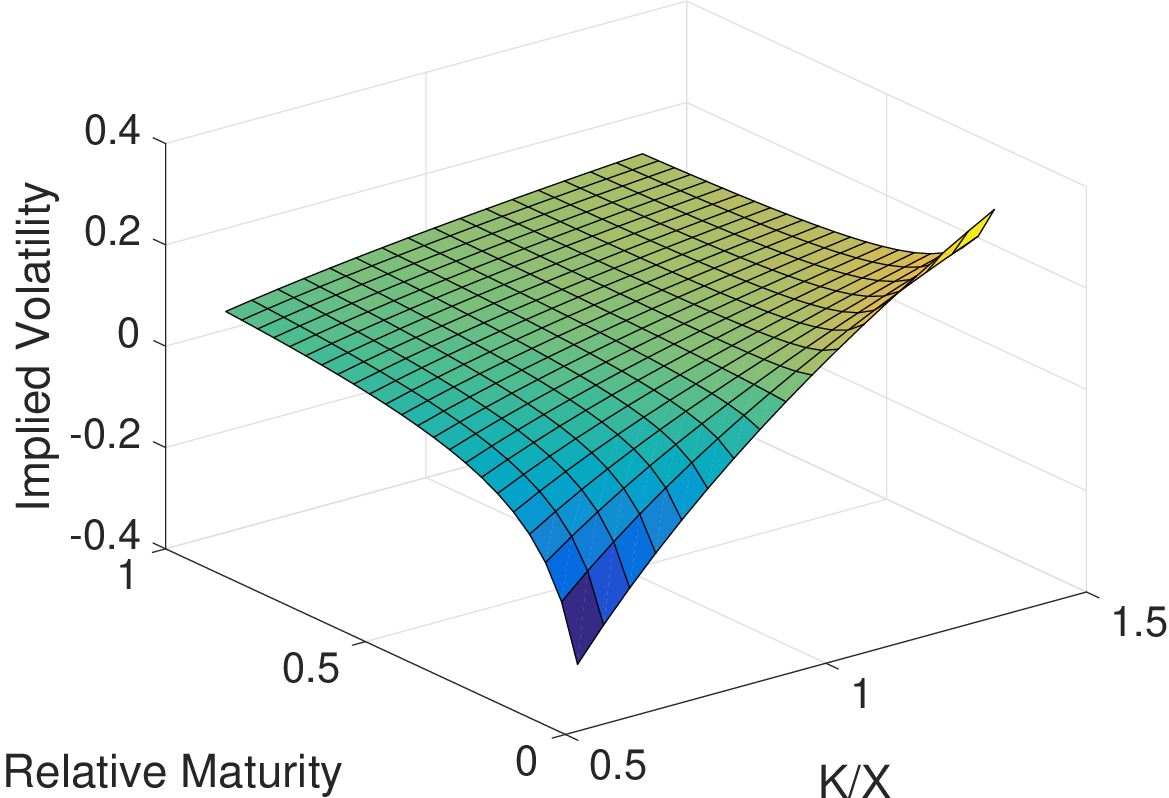}
\end{tabular}
\end{center}
\caption{  The mean implied volatility correction surface as a function of the relative time to maturity  $\tau/\bar\tau$ and the moneyness $K/X$.
The parameters are like those in Figure \ref{fig2_H06}.
\label{fig_surffi}
}
\end{figure}

\begin{proof} 
We   find  by using Eqs. (\ref{pc}) and (\ref{eq:vega})   that
the implied volatility is given by
\begin{equation}
I_t = \overline{\sigma} +     \frac{\phi_t^\eps}{\overline{\sigma}(T-t)}
+\eps^{1-H}  \widetilde{\sigma} \rho D_{t}\Big[ \frac{1}{2\overline{\sigma}(T-t)} + \frac{\log \frac{K}{X_t}}{\overline{\sigma}^3 (T-t)^2} \Big]
+o (\eps^{1-H})  .
\end{equation}
  Because $D_{t}$ is deterministic and given by (\ref{def:DtT}), we can then write
\begin{eqnarray} \label{eq:i3}
I_t & =&  \overline{\sigma} +  \frac{\phi_t^\eps}{\overline{\sigma}(T-t)}    \\ && \nonumber \hbox{}
+ \eps^{1-H} \frac{ \widetilde{\sigma}
 \rho  \left<FF'\right>}{ \overline{\sigma} \Gamma(H+\frac{3}{2})}
 \Big[ \frac{1}{2} (T-t)^{H-\frac{1}{2}}+ \frac{\log \frac{K}{X_t}}{\overline{\sigma}^2 (T-t)^{\frac{3}{2}-H}} \Big]
+o (\eps^{1-H})  
   ,
 \end{eqnarray}
and the Lemma follows.
\end{proof}

 The first two terms in Eq. (\ref{eq:i3}) can be combined and rewritten as (up to terms of order $o(\eps^{1-H})$)
\begin{equation}
\overline{\sigma} +  \frac{\phi_t^\eps}{\overline{\sigma}(T-t)} = \EE\Big[ \frac{1}{T-t} \int_t^T (\sigma_s^\eps)^2 ds \big| {\cal F}_t \Big]^{1/2} 
+o (\eps^{1-H}).
\end{equation}
Because $D_{t}$ is deterministic and given by (\ref{def:DtT}), we can then write
\begin{eqnarray} \nonumber
I_t & =&   \EE\Big[ \frac{1}{T-t} \int_t^T (\sigma_s^\eps)^2 ds \big| {\cal F}_t \Big]^{1/2} 
\\
&&+  \overline\sigma  a_F
\left[  \left( \frac{\tau}{\bar\tau} \right)^{H-1/2} +  \left( \frac{\tau}{\bar\tau} \right)^{H-3/2} {\log \left(\frac{K}{X_t}\right)} \right]
 +o (\eps^{1-H})  ,
\end{eqnarray}
so that the implied volatility is the  expected effective volatility over the remaining time horizon
conditioned on the present and with an added skewness correction.

In view of Eq. (\ref{eq:phisd}), when the  time to maturity is short,
the fourth term (in $\tau^{H-\frac{3}{2}}$) dominates in (\ref{eq:implied2}).
We remark here that this is related to the fact that the small parameter in our problem
is the mean-reversion time, so that for any time to maturity of order one in this regime
the volatility has enough time to fluctuate and mean revert, giving a price correction
  as in Lemma  \ref{lem:1N}.
Moreover, because the ``vega'',  ${\partial_\sigma} C_{\rm BS}$,
is small away from the money (see Eq. (\ref{eq:vega})),  we get
 a strong moneyness dependence, and the implied volatility blows up 
 when the time to maturity goes to zero.

When the time to maturity is long, the third term 
(in $\tau^{H-\frac{1}{2}}$) dominates in (\ref{eq:implied2}).
The long-range dependence gives smooth volatility fluctuations, 
which  gives an implied volatility that blows up when the
time to maturity goes to infinity. The current value of the underlying is less
important in this long-time-to-maturity regime.

\section{The t-T Process and the Stochastic Implied Surface}\label{sec:tT}

We introduced in Eq.~(\ref{def:phit}) the stochastic correction coefficient  
$ \phi_{t}^\eps  \equiv  \phi_{t,T}^\eps $,
 which  gives the random component of the price correction and the implied volatility.
  Note that  we explicitly display here the dependence on maturity $T$. 
If the volatility process had been a  Markovian process, then the correction 
 would have been deterministic,  as in \cite{fouque11}.  The presence of long-range memory  in the volatility
process means that information from the  past (volatility path)  must be carried forward, and
this makes the price correction relative to the price at the homogenized  volatility a stochastic 
process; this is also the case for the implied volatility. 
 
Here we discuss the statistical structure of the random field,  which describes  the implied volatility
surface in the scaling regime that we consider.  
The implied volatility is the central quantity in typical calibration processes.
To design efficient
estimators for both the coherent and incoherent parts  of the implied volatility, 
as well as to characterize
the resulting estimation precision, it is important to understand the statistical 
fluctuations of the observed implied surface.  
We give a precise characterization of these fluctuations below.
The fluctuations of the implied volatility for long times to maturity  (relative to $\bar\tau$)
become strong when the Hurst exponent is large,
because the large Hurst exponent gives strong temporal coherence
and large correction to the anticipated volatility.  
On the other hand, for short times to maturity,
the fluctuations become large when the Hurst exponent is small, because the small Hust exponent gives
a rough process with large fluctuations even over very small intervals.  
 It is also interesting to note that
the correlation structure of the implied volatility surface, in fact, encodes information about
the long-range character of the underlying stochastic volatility. 
Observing, for instance, at-the-money implied volatility fluctuations as a function of current time for fixed time
to maturity gives information that makes it possible to estimate the Hurst exponent
and to check for the consistency of the modeling  framework. 
In \cite{liv}, observed at-the-money implied 
volatility was used to estimate the Hurst exponent. The authors found a coefficient that was slightly
larger than the corresponding estimates using historical data and explained this discrepancy 
in terms of a smoothing effect due to the remaining time to maturity. 
To construct and interpret estimators of this kind,
a model for the implied surface as a random field   relating it  to the 
underlying volatility parameters is clearly essential.

In order to understand the implied volatility random field,
note first  that it follows from  Lemma \ref{lem:3} that as
 $\eps \to 0$, the random process 
 $\eps^{H-1} \phi_{t,T}^\eps  / [\sigma_{\phi}  (T-t)^{H}]  $, $t  <  T$,
converges in distribution (in the sense of finite-dimensional distributions) 
to a Gaussian stochastic  process $\psi_{t,T}$, $t  <  T$, {\it the normalized
t-T correction process},
with mean zero, variance one, and covariance 
$ \EE[\psi_{t,T} \psi_{t',T'}]=   {\cal C}_\phi(t,t';T,T') $ for any 
$t\in [0,T)$, $t ' \in [0,T')$.   
The four-parameter function  ${\cal C}_\phi$ is given by Eq.  (\ref{eq:sCdef}).
We will discuss next in more detail the  t-T process $ \psi_{t,T} $,
a   two-parameter process  of current time $t$ and maturity $T$.
 This process is defined on  $0<t<T$;
 it is a non-stationary Gaussian process, and it is scaled to have constant unit variance. 
As we see below, close to maturity $t\approx T$, the process is strongly affected
by the presence of the  maturity boundary.

Let us first consider the case  of a fixed maturity $T$ and introduce the process
\ba\label{eq:p}
\psi_0(t;T)  = \psi_{t,{T}}  ,\quad \quad t \in [0,T] .
\ea   
When the times are short relative to the time to maturity, i.e. for $|t-t'| \ll     {T}-t$,
   it follows from Eq.~(\ref{eq:sCdef}) that  the  process $(\psi_0(t;T))_{t \in [0,T]}$ 
    decorrelates as
 $$
\EE \big[\psi_0(t;T)\psi_0(t';T)\big]   \sim  1 - \frac{|t-t'|}{2(T-t)},
 $$
which means that it decorrelates as a Markovian process for short times.  
 More generally, the autocovariance function of $(\psi_0(t;T))_{t \in [0,T]}$ is
 \begin{eqnarray*}
 &&\EE \big[\psi_0(t;T)\psi_0(t';T)\big] ={\cal C}(\Delta_0(t,t';T)) , \\
&& {\cal C}(\Delta)  = \frac{\int_{0}^\infty du \big[ \big(u+\frac{|\Delta|+1}{\sqrt{1-\Delta^2}} \big)^{H-\frac{1}{2}}- u ^{H-\frac{1}{2}} \big] 
 \big[ \big( u+\frac{|\Delta|+1}{\sqrt{1-\Delta^2}} \big)^{H-\frac{1}{2}}- \big(u+   \frac{2 |\Delta|}{\sqrt{1-\Delta^2}}   \big)^{H-\frac{1}{2}} \big]}{\int_0^\infty du  \big[ (1+u)^{H-\frac{1}{2}}-u^{H-\frac{1}{2}} \big]^2}   ,
\end{eqnarray*}
with
 \ba
 \label{eq:d0def}
 \Delta_0(t,t';T) =  \frac{t'-t}{|2 {T}-(t+t')|}   ,
 \ea
 which shows that the correlation function of the process $(\psi_0(t;T))_{t \in [0,T]}$  depends only on this relative separation,
giving a situation with a canonical relative decorrelation that depends only on the times to maturity
  $\tau= {T}-t, \tau'= {T}-t'$.  
  Therefore, we introduce  the process $(\psi_1(\tau;T))_{\tau \in [0,T]}$  defined by
\ba
  \psi_1(\tau;T)  =\psi_{T-\tau,T}   ,\quad \quad  \tau \in [0,T] .
\ea
The process $(\psi_1(\tau;T))_{\tau \in [0,T]}$ is Gaussian with mean zero and autocovariance function
$$
 \EE \big[\psi_1(\tau;T)\psi_1(\tau';T)\big] ={\cal C}(\Delta_1(\tau, \tau')) ,
 $$
 with ${\cal C}$ as above and
 \ba
 \label{eq:ddef}
 \Delta_1(\tau, \tau') = \frac{\tau-\tau'}{|\tau+\tau'|}  .
\ea
 For $|\tau-\tau'| \ll \tau$, the process decorrelates on the time scale $\tau$ so that the process fluctuations 
  become more rapid close to maturity. 
Close to maturity, the price  fluctuations become small. When we magnify 
 them, however, we see fluctuations on small time scales when the 
 time to maturity is short,  which reflects  the self-similarity
 of the driving volatility factor. 
 In Figure \ref{fig_corr}, we show the correlation function $\Delta_1\mapsto {\cal C}(\Delta_1)$
 as a function of the relative separation time $\Delta_1\in [-1,1]$ and for $H=0.6$. 
The process decorrelates as a Markovian process for short times; indeed, as one of the
times to maturity goes to zero (relative to the other time to maturity),  the correlation goes rapidly to zero.    

\begin{figure}
\begin{center}
\begin{tabular}{c}
\includegraphics[width=8.4cm]{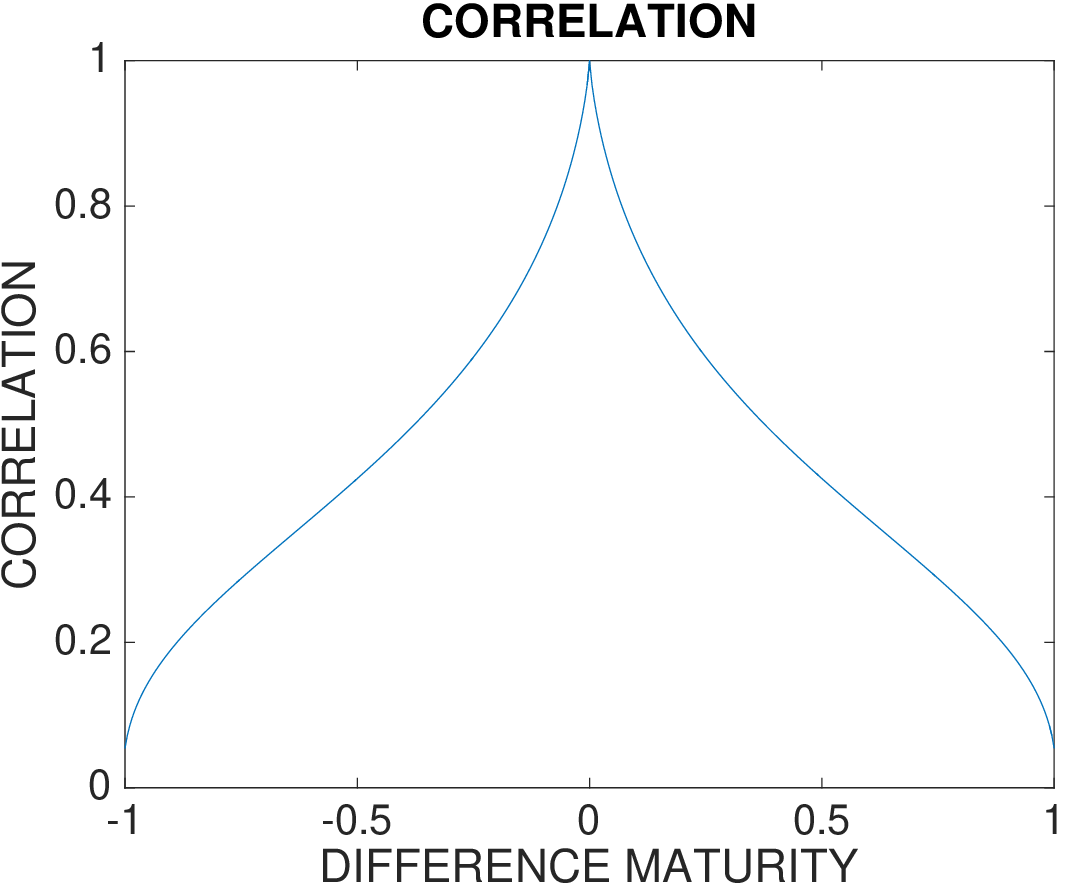}
\end{tabular}
\end{center}
\caption{Autocovariance function of the  t-T process  $\psi_1(\tau;1)$ as a function of the relative time to maturity separation 
   $\Delta_1= (\tau-\tau')/|\tau+\tau'|$ with $H=0.6$. The correlation decays approximately linearly at the origin and  rapidly
   as one of the times  to maturity goes to zero. 
 \label{fig_corr}
}
\end{figure}

Note that it follows from the expression (\ref{eq:ddef}) for $\Delta_1$ that it is {scale invariant}, in that 
$\Delta_1(a \tau, a \tau') = \Delta_1( \tau,  \tau')$ for $a>0$,  giving rapid fluctuations for short times to maturity. 
The process indeed has a self-similar property. We have in distribution
\ban
      \big( \psi_1(\tau;1)  \big)_{\tau \in [0,1]}  \sim     \big(   \psi_1(  \tau T ; T )    \big)_{\tau \in [0,1]}  ,
\ean 
for any $T>0$.
 In Figure \ref{fig_real1}, we show two realizations  of the process   $ \psi_1(\tau;1) $ as a function of time to maturity $\tau$.
 \begin{figure}
\begin{center}
\begin{tabular}{c}
\includegraphics[width=8.4cm]{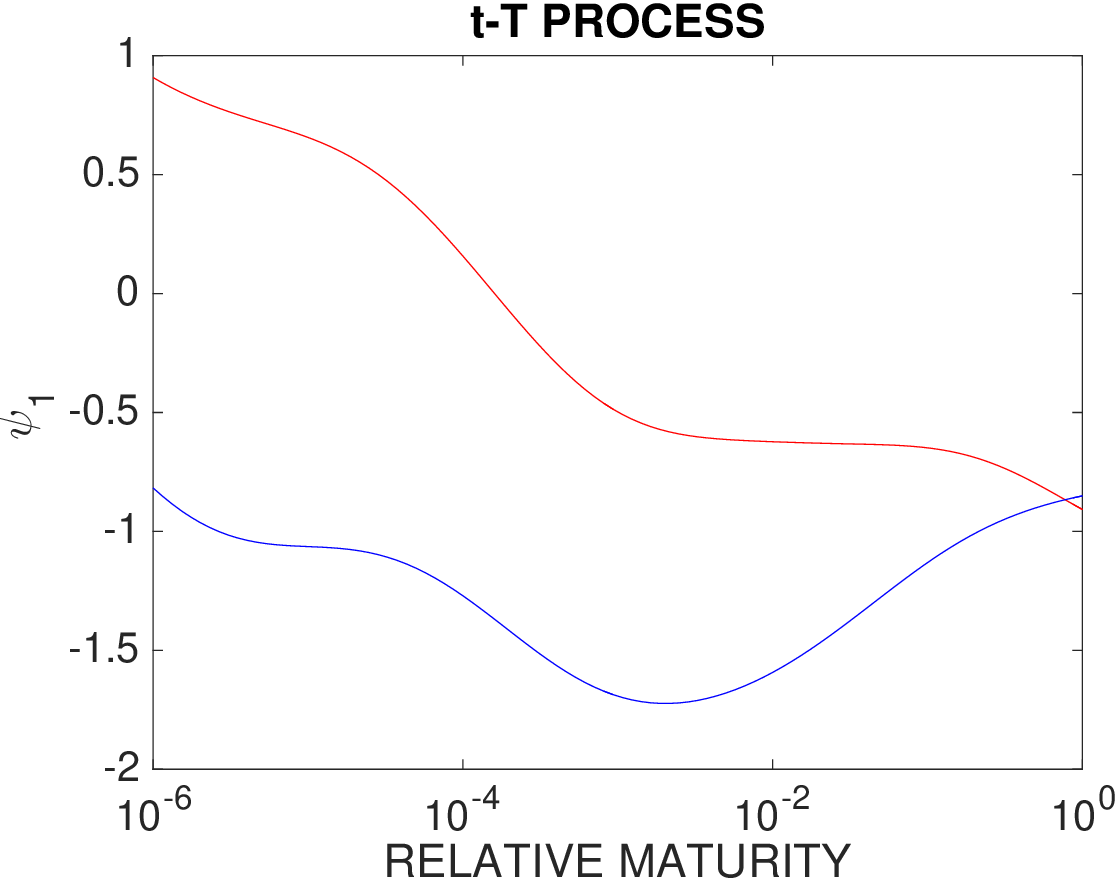}
\end{tabular}
\end{center}
\caption{Realizations of the process $\psi_1(\tau;1)$   as a function of the time to maturity $\tau$ for fixed maturity $T=1$ with $H=0.6$.
 \label{fig_real1}
}
\end{figure}

One can also investigate the structure of the t-T process for a fixed time to maturity $\tau$, as a function of time  $t$.
Thus, if we observe  the price for a given time to maturity, we would like to know how the price correction (and the 
implied volatility)  would  fluctuate with respect to  the current time, or time translation.
Accordingly, we consider the process
\ba
\psi_2(t;\tau)  = \psi_{t, {\tau+t}}  ,\quad \quad t \geq 0,
\ea   
for fixed $\tau>0$.  
The process $(\psi_2(t;\tau))_{t \in [ 0 ,\infty)}$ is Gaussian with mean zero and autocovariance function
\begin{eqnarray}
\label{eq:ssCdef5}
&&   \EE \big[\psi_2(t;\tau)\psi_2(t';\tau)\big] ={\cal C}_2( \Delta_2(t,t';\tau)) ,\\ 
&& {\cal C}_2( \Delta )= \frac{\int_{0}^\infty du \big[ (u+1 )^{H-\frac{1}{2}}- u ^{H-\frac{1}{2}} \big] 
 \big[ (u+1+|\Delta| )^{H-\frac{1}{2}}-(u+|\Delta| )^{H-\frac{1}{2}} \big]}{\int_0^\infty du 
  \big[ (1+u)^{H-\frac{1}{2}}-u^{H-\frac{1}{2}} \big]^2} ,
\nonumber
\end{eqnarray}
with
 \ba
 \Delta_2(t,t';\tau) =\frac{t'-t}{\tau}      .
 \ea    
 The expression of $\Delta_2$ shows that the  coherence time  of this  process  is proportional to the time to maturity $\tau$. 
We see again that the rescaled implied volatility surface fluctuations are more rapid when they are close to maturity.
 We also see that  on transects parallel to the maturity boundary 
 in the $t,T$ plane, these fluctuations are stationary. 
 This is consistent with the fact that we have  an underlying consistent model with a stationary volatility driving factor. 
   The fluctuations, moreover, have  a self-similar property. We have in distribution
\ban
      \big( \psi_2(t;1)  \big)_{t \in [0,\infty)}  \sim     \big(   \psi_2(  \tau t ; \tau )    \big)_{t \in [0,\infty)}  ,
\ean 
for any $\tau>0$.
The autocovariance function of $(\psi_2(t;1))_{t \in [0,\infty)} $ is plotted in Figure \ref{fig_corr2}.
 In the figure, one can see the rapid decay at the origin followed by a long-range behavior.  This shows how
 the implied  surface decorrelates as we move in time.  In Figure \ref{fig_corr3}, we show
 the autocorrelation function  in a {\it log-log} plot with the dashed line corresponding to the  
correlation decay $|t'-t|^{2H-2}$. 
In Figure \ref{fig_real2}, we show two realizations  of  the process $\psi_2(t;1)$.

 \begin{figure}
\begin{center}
\begin{tabular}{c}
\includegraphics[width=8.4cm]{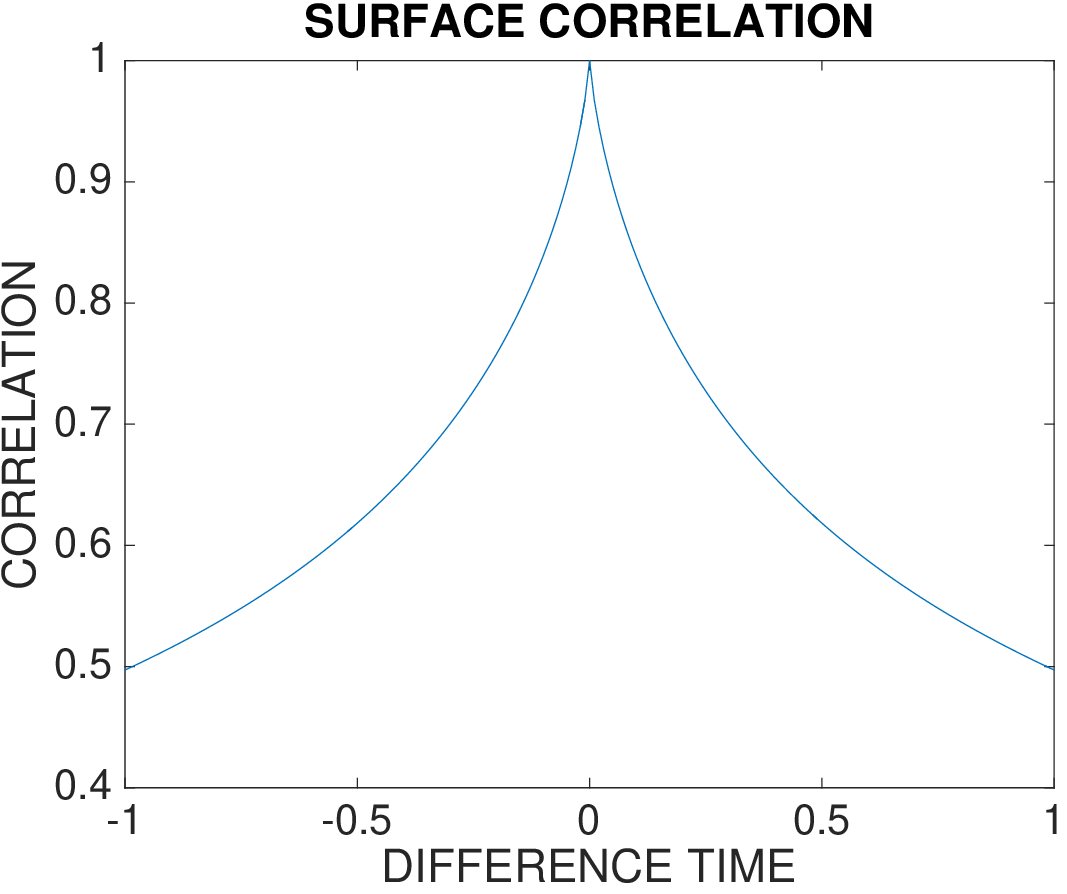}
\end{tabular}
\end{center}
\caption{ Autocovariance function of the  t-T process  $\psi_2(t;1)$ as a function of the time $t'-t$
for fixed time to maturity $\tau=1$ 
 with $H=0.6$.   On the short time scales, the process decorrelates as a Markovian process;  on the long 
 time scales, it exhibits long-range correlations.  
 \label{fig_corr2}
}
\end{figure}
\begin{figure}
\begin{center}
\begin{tabular}{c}
\includegraphics[width=8.4cm]{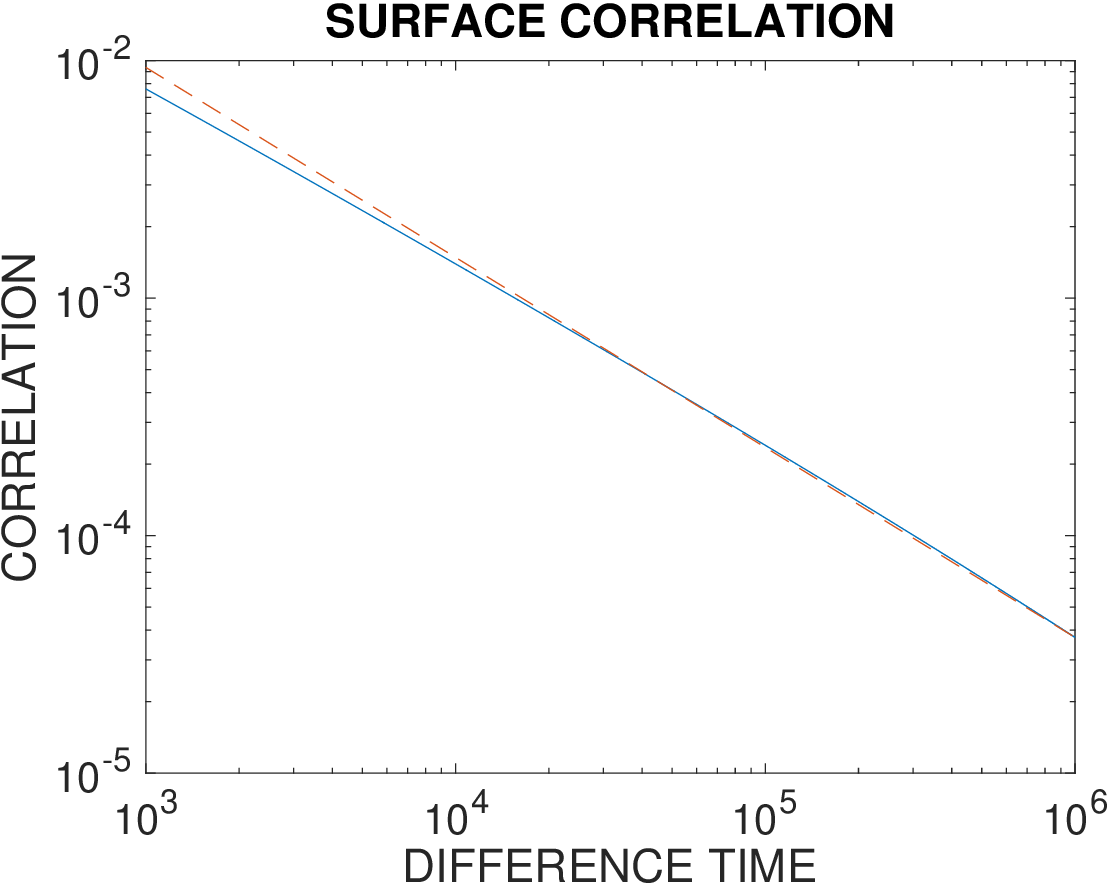}
\end{tabular}
\end{center}
\caption{ Autocovariance function of the  t-T process  $\psi_2(t;1)$  as in Figure \ref{fig_corr2}, 
but on a {\it log-log} scale with the dashed  line showing  the   decay $|t'-t|^{2H-2}$.
 \label{fig_corr3}
}
\end{figure}
  \begin{figure}
\begin{center}
\begin{tabular}{c}
\includegraphics[width=8.4cm]{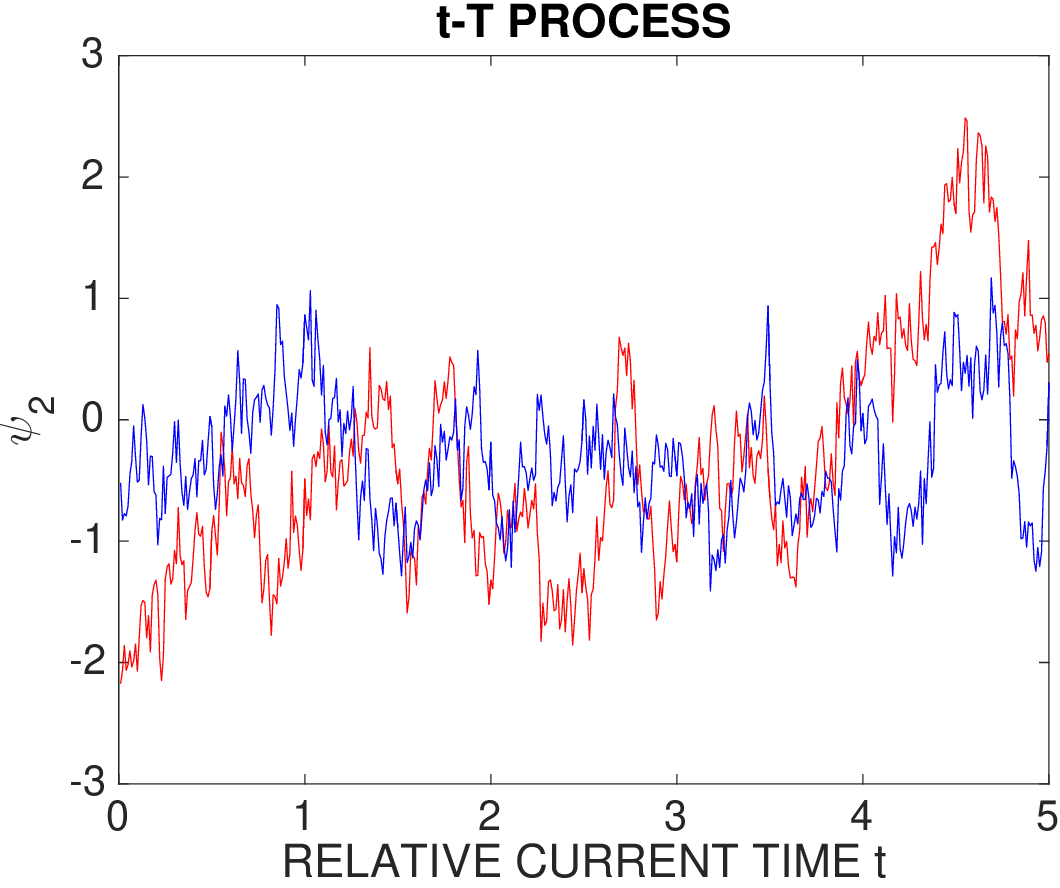}
\end{tabular}
\end{center}
\caption{Realizations of the process $\psi_2(t;1)$ with $H=0.6$.
 \label{fig_real2}
}
\end{figure}

Finally, it is of interest to consider the case where we evaluate the stochastic correction factor
as a function of time to maturity for the fixed current time $t$,
\ba
\psi_3(\tau ; t)=     \psi_{t,{t+\tau}} ,\quad \quad \tau \geq 0  .
\ea   
 The process $(\psi_3(\tau;t))_{\tau \in [0,\infty) }$ is Gaussian with mean zero and autocovariance function
\begin{eqnarray*}
&&   \EE \big[\psi_3(\tau;t)\psi_3(\tau';t)\big]  ={\cal C}_3( \Delta_3(\tau,\tau') ) , \\
&& {\cal C}_3( \Delta)=  
  \frac{\int_{0}^\infty du \big[ (u+1/\sqrt{1+|\Delta|})^{H-\frac{1}{2}}- u ^{H-\frac{1}{2}} \big] 
 \big[ (u+\sqrt{1+|\Delta|} )^{H-\frac{1}{2}}-u^{H-\frac{1}{2}} \big]}{\int_0^\infty du  \big[ (1+u)^{H-\frac{1}{2}}-u^{H-\frac{1}{2}} \big]^2} ,
  \end{eqnarray*}
 with
\ba
\label{eq:pp}
 \Delta_3(\tau,\tau') =    \frac{\tau-\tau'}{  \tau  \wedge \tau' }     .
 \ea    
This covariance function is plotted in Figure \ref{fig_corr4}. 
 Note that it follows from the expression (\ref{eq:pp}) for $\Delta_3$   that it is {scale invariant} in that 
$\Delta_3(a\tau, a\tau') = \Delta_3( \tau,  \tau')$ for $a>0$,  so that again the process
fluctuates more rapidly for small maturities.
The distribution of the process $(\psi_3(\tau;t))_{\tau\in [0,\infty)}$ does not depend on $t$, and 
it has a self-similar property. For any $a>0$, 
we have in distribution
\ban
     \big(  \psi_3(\tau;t)   \big)_{\tau \in [0,\infty)} \sim      \big(  \psi_3(a \tau;t)  \big)_{\tau \in [0,\infty)}     .
\ean 
In Figure \ref{fig_real3}  we show two  realizations of the process $(\psi_3(\tau;t))_{\tau \in [0,1)} $.  
 \begin{figure}
\begin{center}
\begin{tabular}{c}
\includegraphics[width=8.4cm]{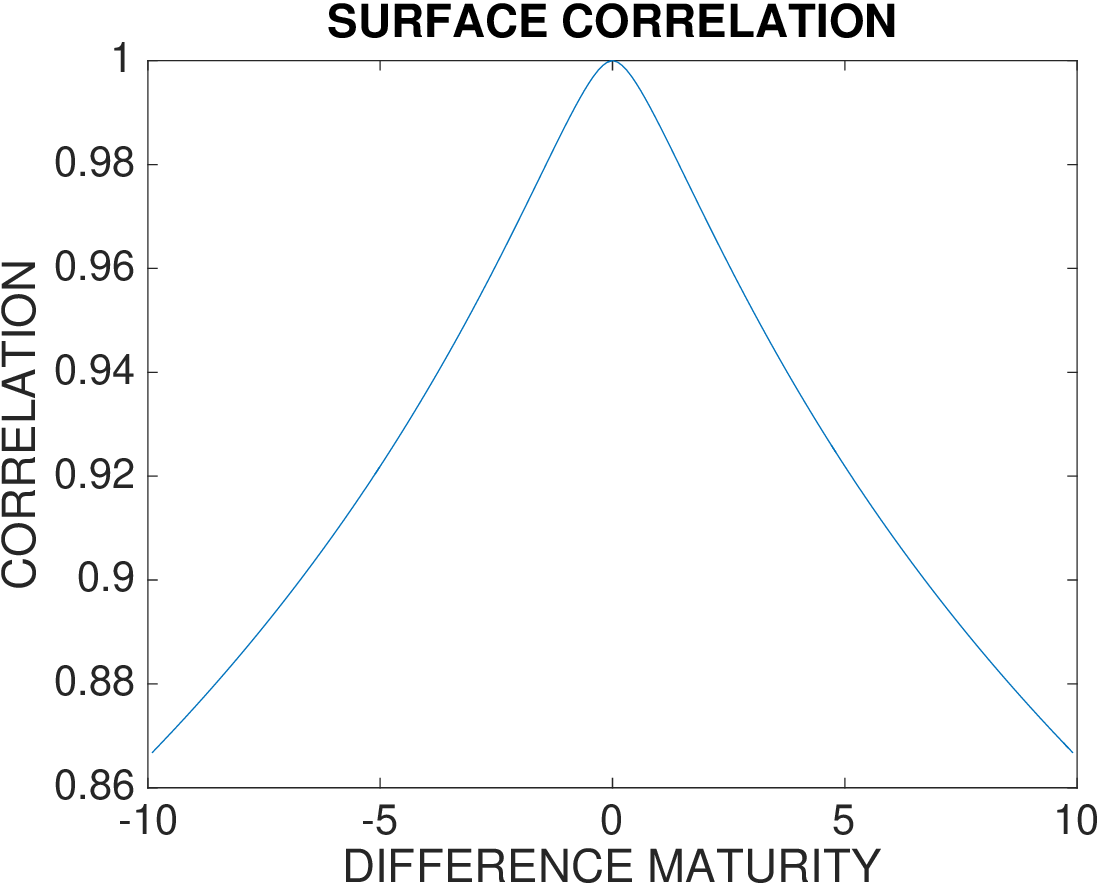}
\end{tabular}
\end{center}
\caption{  Autocovariance function of the  t-T process  $\psi_3(\tau;1)$  as a function of  the relative  time to maturity  separation
   $\Delta_3=(\tau-\tau')/(\tau\wedge\tau')$ with $H=0.6$. Note that the correlation function  exhibits slow decay.
    \label{fig_corr4}
}
\end{figure}
\begin{figure}
\begin{center}
\begin{tabular}{c}
\includegraphics[width=8.4cm]{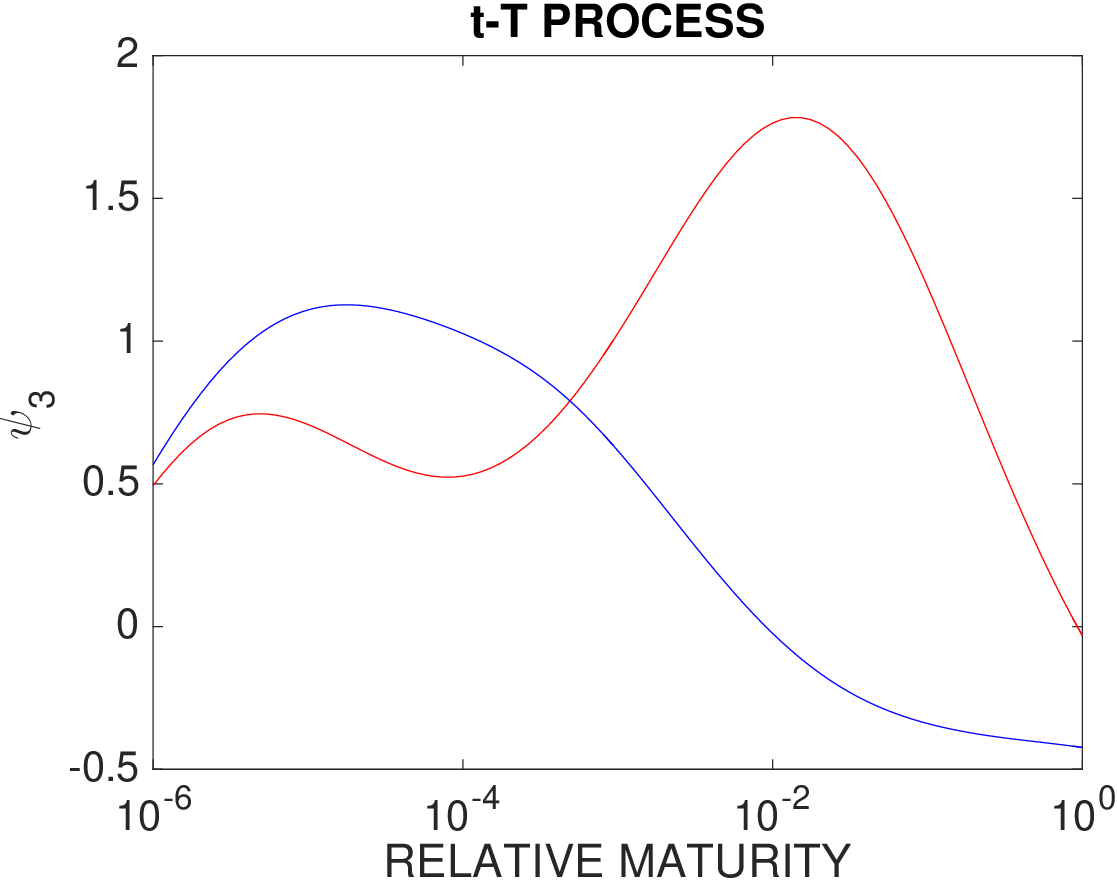}
\end{tabular}
\end{center}
\caption{    Realizations of the process $\psi_3(\tau;1)$  for fixed current time $t=1$ and $H=0.6$, 
with the smooth and slow decay of the correlations giving a smooth time-to-maturity dependence. 
 \label{fig_real3}
}
\end{figure}

\section{Conclusion}
\label{sec:conclusion}
We have considered a continuous time stochastic volatility model with long-range correlation properties.
We have addressed the regime of fast mean reversion. This makes it possible to derive
an explicit expression for the approximate European call option price and the implied volatility.
Specifically the volatility is a smooth function of a fractional Ornstein--Uhlenbeck 
process. 
Analysing such a non-Markovian situation is challenging. To the 
best of our knowledge, we present the first analytical expression for the 
price approximation for general maturities when the volatility fluctuations are of order one. So far
the price computations for such situations have been based on numerical 
approximations.  The main result from the applied view point is then the form
of the fractional term structure that we obtain for the implied volatility surface.
Indeed, we get an implied volatility that grows large with time to maturity while
generating a strong skew for short times to maturity, which is consistent with common observations.
We stress that in our formulation,  the mean-reversion time is small compared to any
fixed maturity as we consider a fast mean reverting process.  Let us note, finally, that we 
have considered the case of processes with long-range correlation properties with the Hurst exponent
$H>1/2$ explaining the  growth of implied volatility for large maturity.

\section*{Acknowledgements}
This paper is based upon research supported in part by Centre Cournot, Fondation Cournot, and Universit\'e Paris Saclay
(chaire D'Alembert).

\bigskip

\begin{appendices}

\section{Hermite Decomposition of the Stochastic Volatility Model}
\label{sec:hermite}
We denote
\begin{equation}
\label{def:Ftilde}
\widetilde{F}(z) = F(\sigma_{{\rm ou}} z)^2 .
\end{equation}
Because $\EE[\widetilde{F}(Z)^2]<\infty$ is finite when $Z$ is a standard normal variable,
the function $\widetilde{F}$ can be expanded in terms of the Hermite polynomials
\begin{equation}
H_k(z) = (-1)^k e^{z^2/2} \frac{d^k}{dz^k} e^{-z^2/2} ,
\end{equation}
and the series
\begin{equation}
\sum_{k=0}^\infty \frac{C_k}{k!} H_k(z) ,
\end{equation}
with
\begin{equation}
C_k = \EE \big[ H_k(Z) \widetilde{F}(Z) \big] =  \int_\RR H_k(z) \widetilde{F}(z) p(z) dz,
\label{def:Ck}
\end{equation}
converges in $L^2(\RR, p(z)dz)$ to $\widetilde{F}(z)$.
The Hermite polynomials satisfy
$$
\EE[H_k(Z) H_j(Z)] = \int_\RR H_k(z) H_j(z) p(z) dz = \delta_{kj} k!   ,
$$
and we have $\sum_{k=0}^\infty \frac{C_k^2}{k!} = \EE[\widetilde{F}(Z)^2]<\infty$.
Note that $C_0= \left< F^2 \right>$.

\begin{lemma}
\label{lem:a}
If there exists $\alpha >2$ such that the function $\widetilde{F}$ defined by (\ref{def:Ftilde}) satisfies
\begin{equation}
\label{hyp:H}
\sum_{k=0}^\infty \frac{\alpha^k C_k^2}{k!} < \infty ,
\end{equation}
then the random process
\begin{equation}
\label{def:Iepst}
I_t^\eps = \int_0^t F^2( {Z}^\eps_s) - \left< F^2 \right>ds
\end{equation}
satisfies
\begin{equation}
\label{eq:boundIepst}
\sup_{t \in [0,T]} \EE[ (I^\eps_t)^4 ] \leq K \eps^{4-4H} ,
\end{equation}
for some constant $K$.
\end{lemma}

\begin{proof}
Denoting $\widetilde{Z}^\eps_t =\sigma_{{\rm ou}}^{-1} {Z}^\eps_t$,
which is a zero-mean Gaussian process with covariance function
$\EE [ \widetilde{Z}^\eps_t \widetilde{Z}^\eps_{t+s}] = {\cal C}_Z(s/\eps)$,
we have
$$
I^\eps_t
= 
 \int_0^t \widetilde{F}(\widetilde{Z}^\eps_s) - \left< F^2 \right>ds =\sum_{m=1}^\infty C_m I_{t,m}^\eps  ,
$$
where
$$
I_{t,m}^\eps = \frac{1}{m!} \int_0^t H_m(\widetilde{Z}^\eps_s) ds ,\quad \quad m\geq 1.
$$
From \cite[Lemma 2.2]{taqqu78}, the fourth-order moment of $I_{t,m}^\eps$ can be expanded as
$$
\EE [ (I_{t,m}^\eps)^4 ] = \frac{1}{2^m (2m)!} \sum \int_0^t \cdots \int_0^t dt_1 dt_2 dt_3 dt_4 
\prod_{\ell=1}^m {\cal C}_Z\Big(\frac{t_{i_\ell}-t_{j_\ell}}{\eps}\Big)  ,
$$
where the sum is over all indices $i_1,j_1,\ldots,i_{2m},j_{2m}$ such that:\\
i) $i_1,j_1,\ldots,i_{2m},j_{2m} \in \{1,2,3,4\}$,\\
ii) $i_1\neq j_1$, $\ldots$, $i_{2m}\neq j_{2m}$,\\
iii) each number $1,2,3,4$ appears exactly $m$ times in $(i_1,j_1,\ldots,i_{2m},j_{2m})$.\\
The number $N_{2m}$ of terms in this sum is, therefore, smaller than $(4m)! / m!^4$ (it would be exactly this cardinal without the second condition;
therefore it is smaller than this number).\\
Because ${\cal C}_Z(s) \leq 1 \wedge K |s|^{2H-2}$ for some constant $K$, we have, for any $t \in [0,T]$,
$$
\EE [ (I_{t,m}^\eps)^4 ] \leq \frac{1}{2^{2m} (2m)!} \sum \int_0^T \cdots \int_0^T dt_1 dt_2 dt_3 dt_4
\prod_{\ell=1}^{2m}
1\wedge K \big(\frac{|t_{i_\ell}-t_{j_\ell}|}{\eps}\big)^{2H-2} .
$$
For each term of the sum, we apply the change of variables $s_1=t_{i_1}$, $s_2=t_{j_1}$, 
$s_3=t_{\min( \{1,2,3,4\} \backslash \{i_1,j_1\})}$,
$s_4=t_{\max( \{1,2,3,4\} \backslash \{i_1,j_1\})}$.
In the product, we keep the first term: $K(|s_1-s_2|/\eps)^{2H-2}$,
and the first term that has $s_3$ in it: $K(|s_3-s_j|/\eps)^{2H-2}$, so that
we can write, for any $t \in [0,T]$,
\begin{eqnarray*}
\EE [ (I_{t,m}^\eps)^4 ] &\leq& \frac{N_{2m} K^2}{2^{2m} (2m)!}  \int_0^T \cdots \int_0^T ds_1 ds_2 ds_3 ds_4
\big(\frac{|s_1-s_2|}{\eps}\big)^{2H-2} 
\Big[ \big(\frac{|s_3-s_1|}{\eps}\big)^{2H-2} \\
&&\quad \quad \quad \quad +
 \big(\frac{|s_3-s_2|}{\eps}\big)^{2H-2} 
+
 \big(\frac{|s_3-s_4|}{\eps}\big)^{2H-2} \Big] \\
  &\leq& K' \frac{(4m)!}{2^{2m} (2m)!m!^4} \eps^{4-4H}  ,
\end{eqnarray*}
for some constant $K'$ (that depends on $H$ and $T$), because $s^{2H-2}$ is integrable over $[0,T]$.
By Stirling's formula, we obtain
$$
\frac{(4m)!}{2^{2m} (2m)!m!^4} \simeq \frac{2^{2m}}{m!^2} \frac{1}{\sqrt{2}\pi m} .
$$
Therefore, by Minkowski's inequality, we have, for any $t \in [0, T]$,
\begin{eqnarray*}
\EE [ (I_t^\eps)^4) ]^{1/4} &\leq& \sum_{m=1}^\infty |C_m| \EE [ (I^\eps_m)^4) ]^{1/4}
\leq K'' \eps^{1-H} \sum_{m=1}^\infty |C_m| \Big(\frac{2^{m}}{m!}\Big)^{1/2} 
\\
&\leq& 
K'' \eps^{1-H} \Big( \sum_{m=1}^\infty \frac{\alpha^m C_m^2}{m!}\Big)^{1/2} \Big( \sum_{m=1}^\infty \frac{2^m}{\alpha^m}\Big)^{1/2} ,
\end{eqnarray*}
for some constant $K''$,
which gives the desired result.
\end{proof}

The hypothesis (\ref{hyp:H}) in Lemma \ref{lem:a} requires some smoothness for the function $\widetilde{F}$.
The following lemma gives a sufficient condition.

\begin{lemma}
If the function $\widetilde{F}$ defined by (\ref{def:Ftilde}) is of the form
\begin{equation}
\widetilde{F}(x) = \int_{-\infty}^x f(y) dy ,
\end{equation}
where the Fourier transform of the function $f$ satisfies $|\hat{f}(\nu)| \leq C \exp(-\nu^2)$ for some $C>0$,
then there exists $K >0$ such that, for any $k\geq 0$,
\begin{equation}
\label{hyp:H2}
\frac{C_k^2}{k!} \leq K 3^{-k}.
\end{equation}
\end{lemma}
The inequality (\ref{hyp:H2})
is sufficient to ensure that the hypothesis (\ref{hyp:H}) is fulfilled.
We may, for instance, consider
\begin{equation}
\widetilde{F}(x) =  \int_{-\infty}^{x} e^{-y^2/4} dy  \mbox{ or }
\widetilde{F}(x) =  \int_{-\infty}^{x} {\rm sinc}^2(y) dy .
\end{equation}

\begin{proof}
The function $\widetilde{F}$ is of class ${\cal C}^\infty$, and
we have, for any $k\geq 1$, using integration by parts,
$$
C_ k = \int_\RR \widetilde{F}(z) H_k(z) p(z) dz =  \int_\RR \widetilde{F}^{(k)}(z) p(z) dz =  \int_\RR {f}^{(k-1)}(z) p(z) dz .
$$
By Parseval formula, we have
$$
C_k= \frac{1}{2\pi}
\int_\RR e^{-\nu^2/2} (i\nu)^{k-1} \hat{f}(\nu) d\nu .
$$
Because $|\hat{f}(\nu)| \leq C \exp(-\nu^2)$, we obtain
$$
|C_k|\leq C \int_\RR e^{-3\nu^2/2} |\nu|^{k-1}  d\nu = C \Big( \frac{2}{3} \Big)^{\frac{k}{2}}
\int_0^\infty e^{-s} s^{\frac{k}{2}-1} ds = C\Big( \frac{2}{3} \Big)^{\frac{k}{2}} \Gamma\Big(\frac{k}{2}\Big) ,
$$
which gives the desired result using Stirling's formula $\Gamma(z) \sim z^{z-1/2} e^{-z} \sqrt{2\pi}$.
\end{proof}

\section{Technical Lemmas}\label{sec:app}
We denote
\begin{equation}
\label{def:G}
G(z) = \frac{1}{2} \big( F(z)^2 - \overline{\sigma}^2\big) .
\end{equation}
The martingale $\psi^\eps_t$ defined by (\ref{def:Kt}) has the form
\begin{equation}
\psi_t^\eps = 
\EE \Big[   \int_0^T G(Z_s^\eps)  ds \big| {\cal F}_t\Big] .
\end{equation}

\begin{lemma}
\label{lem:1}
$(\psi_t^\eps)_{t\in [0,T]}$ is a square-integrable martingale and
\begin{equation}
d \left< \psi^\eps, W\right>_t =  \vartheta^\eps_{t} dt ,
\quad \quad \vartheta^\eps_{t} = \sigma_{{\rm ou}} \int_t^T \EE \big[ G'(Z_s^\eps)|{\cal F}_t \big]{\cal K}^\eps(s-t) ds   .
\end{equation}
\end{lemma}
An alternative expression of the bracket $\left< \psi^\eps, W\right>_t$ is given in (\ref{eq:crocKW1}-\ref{eq:crocKW2}).

\begin{proof}
For $t\leq s$, the conditional distribution of $Z_s^\eps$ given ${\cal F}_t$ is Gaussian with mean
$$
\EE \big[  Z_s^\eps |{\cal F}_t \big] =\sigma_{{\rm ou}}
 \int_{-\infty}^t {\cal K}^\eps(s-u) dW_u 
$$
and deterministic variance given by
$$
 {\rm Var} \big( Z_s^\eps |{\cal F}_t\big) = (\sigma_{0,s-t}^\eps)^2,
$$
where we have defined, for any $0\leq t\leq s\leq \infty$,
\begin{equation}
(\sigma_{t,s}^\eps)^2=\sigma_{{\rm ou}}^2
   \int_t^{s} {\cal K}^\eps(u)^2 du .
\end{equation}
 We thus  have  that  the distribution of
$$
\frac{1}{\sigma_{0,s-t}^\eps} \Big( \big( Z_s^\eps  - 
\sigma_{{\rm ou}} 
\int_{-\infty}^t {\cal K}^\eps(s-u) dW_u  \big)   \big| {\cal F}_t 
 \Big) 
  $$
is standard normal.
Therefore, we have
$$
\EE \big[ {G}( Z_s^\eps ) |{\cal F}_t \big] 
= \int_\RR G \Big( \sigma_{{\rm ou}} \int_{-\infty}^t {\cal K}^\eps(s-u) dW_u
+\sigma_{0,s-t}^\eps z\Big) p(z) dz ,
$$
where $p(z)$ is the pdf of the standard normal distribution. 
As a random process in $t$, it is a continuous martingale.
By It\^o's formula, for any $t\leq s$,
\begin{eqnarray*}
\EE \big[ {G}( Z_s^\eps ) |{\cal F}_t \big] &=&
 \int_\RR {G} \Big( \sigma_{{\rm ou}}\int_{-\infty}^0 {\cal K}^\eps(s-v) dW_v
+\sigma_{0,s}^\eps z\Big) p(z) dz  \\
&&
+ \int_0^t  \int_\RR {G}' \Big( \sigma_{{\rm ou}}\int_{-\infty}^u {\cal K}^\eps(s-v) dW_v
+\sigma_{0,s-u}^\eps z\Big) z p(z) dz \partial_u \sigma_{0,s-u}^\eps du   \\
&&
+\sigma_{{\rm ou}} \int_0^t  \int_\RR {G}' \Big( \sigma_{{\rm ou}}\int_{-\infty}^u {\cal K}^\eps(s-v) dW_v
+\sigma_{0,s-u}^\eps z\Big)p(z) dz  {\cal K}^\eps(s-u) dW_u   \\
&&
+\frac{\sigma_{{\rm ou}}^2}{2}  \int_0^t  \int_\RR {G}'' \Big( \sigma_{{\rm ou}}\int_{-\infty}^u {\cal K}^\eps(s-v) dW_v
+\sigma_{0,s-u}^\eps z\Big) p(z) dz {\cal K}^\eps(s-u)^2 du 
\end{eqnarray*}
and
\begin{eqnarray*}
{G}( Z_s^\eps )  &=&{G} \Big( \sigma_{{\rm ou}}\int_{-\infty}^s {\cal K}^\eps(s-v) dW_v\Big)  \\
 &=&  \int_\RR {G} \Big(\sigma_{{\rm ou}} \int_{-\infty}^s {\cal K}^\eps(s-v) dW_v
+\sigma_{0,0}^\eps z\Big) p(z) dz  \\
&=&  \int_\RR {G} \Big(\sigma_{{\rm ou}} \int_{-\infty}^0 {\cal K}^\eps(s-v) dW_v
+\sigma_{0,s}^\eps z\Big) p(z) dz  \\
&&
+ \int_0^s \int_\RR {G}' \Big( \sigma_{{\rm ou}}\int_{-\infty}^u {\cal K}^\eps(s-v) dW_v
+\sigma_{0,s-u}^\eps z\Big) z p(z) dz \partial_u \sigma_{0,s-u}^\eps du   \\
&&
+\sigma_{{\rm ou}}\int_0^s  \int_\RR {G}' \Big( \sigma_{{\rm ou}}\int_{-\infty}^u {\cal K}^\eps(s-v) dW_v
+\sigma_{0,s-u}^\eps z\Big)p(z) dz  {\cal K}^\eps(s-u) dW_u   \\
&&
+\frac{\sigma_{{\rm ou}}^2}{2}  \int_0^s  \int_\RR {G}'' \Big(\sigma_{{\rm ou}} \int_{-\infty}^u {\cal K}^\eps(s-v) dW_v
+\sigma_{0,s-u}^\eps z\Big) p(z) dz {\cal K}^\eps(s-u)^2 du  .
\end{eqnarray*}
Therefore,
\begin{eqnarray*}
\psi^\eps_t &=& \int_0^t  {G}( Z_s^\eps) ds + \int_t^T \EE \big[ {G}( Z_s^\eps ) |{\cal F}_t \big]  ds \\
&=& 
\Big[ \int_\RR \int_0^T {G} \Big( \sigma_{{\rm ou}} \int_{-\infty}^0 {\cal K}^\eps(s-v) dW_v
+\sigma_{0,s}^\eps z\Big)  ds p(z) dz\Big] \\
&& +\int_0^t \Big[ \int_u^T  \int_\RR {G}' \Big(\sigma_{{\rm ou}} \int_{-\infty}^u {\cal K}^\eps(s-v) dW_v
+\sigma_{0,s-u}^\eps z\Big) z p(z) dz \partial_u \sigma_{0,s-u}^\eps ds \Big] du \\
&&+
\sigma_{{\rm ou}}\int_0^t \Big[ \int_u^T \int_\RR {G}'\Big( \sigma_{{\rm ou}}\int_{-\infty}^u {\cal K}^\eps(s-v) dW_v
+\sigma_{0,s-u}^\eps z\Big)p(z) dz  {\cal K}^\eps(s-u) ds \Big]dW_u \\
&& +\frac{\sigma_{{\rm ou}}^2}{2} \int_0^t \Big[ \int_u^T \int_\RR {G}'' \Big(\sigma_{{\rm ou}} \int_{-\infty}^u {\cal K}^\eps(s-v) dW_v
+\sigma_{0,s-u}^\eps z\Big) p(z) dz {\cal K}^\eps(s-u)^2 ds \Big] du .
\end{eqnarray*}
This gives
\begin{eqnarray}
\label{eq:crocKW1}
d \left< \psi^\eps, W\right>_t =\vartheta^\eps_ t  dt ,
\end{eqnarray}
with
\begin{eqnarray}
\vartheta^\eps_t  &=& \sigma_{{\rm ou}}
 \int_t^T \int_\RR {G}'\Big( \sigma_{{\rm ou}}\int_{-\infty}^t {\cal K}^\eps(s-v) dW_v
+\sigma_{0,s-t}^\eps z\Big)p(z) dz  {\cal K}^\eps(s-t) ds  ,
\label{eq:crocKW2}
\end{eqnarray}
which can also be written 
as stated in the Lemma.
\end{proof}

The important properties of the random process $\vartheta^\eps_{t}$ are stated in the following lemma.
\begin{lemma}
\label{lem:2}
For any $t \in [0,T]$, we have
\begin{equation}
\vartheta^\eps_{t} = \eps^{1-H}  \theta_{t} + \widetilde{\theta}^\eps_{t}  ,
\end{equation}
where $\theta_{t} $ is deterministic and defined by
\begin{equation}
\theta_{t} = \overline{\theta} (T-t)^{H -\frac{1}{2}} ,\quad \quad \overline{\theta} = \frac{
\left< G' \right>}{\Gamma(H+\frac{1}{2})}
,
\end{equation}
and $\widetilde{\theta}^\eps_{t}$ is random, but smaller than $\eps^{1-H}$,
\begin{equation}\label{eq:phit}
\limsup_{\eps \to 0} \eps^{H-1} \sup_{t \in [0,T]} \EE \big[ (\widetilde{\theta}^\eps_{t})^2\big]^{1/2} =0.
\end{equation}
\end{lemma}

\begin{proof}
Recall first from Eq. (\ref{eq:kdef}) that
\begin{equation*}
{\cal K}^\eps(t) = \frac{1}{\sqrt{\eps}} {\cal K}\Big(\frac{t}{\eps}\Big),\quad \quad
{\cal K}(t) =\frac{1}{
\sigma_{{\rm ou}}
\Gamma(H+\frac{1}{2})
} 
 \Big[ t^{H - \frac{1}{2}} - \int_0^t (t-s)^{H - \frac{1}{2}} e^{-s} ds \Big] .
\end{equation*} 
The expectation of $  \vartheta^\eps_{t}  $ is then equal to
$$
 \EE\big[ \vartheta^\eps_{t}  \big ] 
= \sigma_{{\rm ou}}  \left< G'\right>  \int_0^{T-t} {\cal K}^\eps(s) ds
= \sigma_{{\rm ou}}  \left< G'\right> \sqrt{\eps} \int_0^{(T-t)/\eps} {\cal K}(s) ds.
$$
Therefore, the difference
$$
 \EE\big[ \vartheta^\eps_{t}  \big ]  - \eps^{1-H} \theta_{t} 
=
\sigma_{{\rm ou}}  \left< G'\right> {\eps}^{1/2} \int_0^{(T-t)/\eps} {\cal K}(s) - \frac{s^{H-\frac{3}{2}}}{
\sigma_{{\rm ou}}
\Gamma(H-\frac{1}{2})} ds
$$
can be bounded by
\begin{equation}
\big| \EE\big[ \vartheta^\eps_{t}  \big ]   - \eps^{1-H} \theta_{t}  \big| 
\leq 
C  \eps^{1/2} ,
\end{equation}
uniformly in $t\in [0,T]$, for some constant $C$,
because ${\cal K}(s) - \frac{s^{H-\frac{3}{2}}}{
\sigma_{{\rm ou}}
\Gamma(H-\frac{1}{2})} $ is in $L^1$.

We have
\begin{eqnarray*}
{\rm Var}(\vartheta^\eps_{t}) &=& 
  \sigma_{{\rm ou}}^2 \int_t^T ds\int_t^T ds'  {\cal K}^\eps(s-t) {\cal K}^\eps(s'-t)  
{\rm Cov}\big(  \EE \big[ G'(Z_s^\eps) |{\cal F}_t \big], \EE \big[ G'(Z_{s'}^\eps) |{\cal F}_{t} \big]\big) \\
&\leq& \sigma_{{\rm ou}}^2 \Big( \int_t^T ds {\cal K}^\eps(s-t)   {\rm Var} \big( \EE \big[ G'(Z_s^\eps) |{\cal F}_t \big]\big)^{1/2}
 \Big)^{2}   \\
&= &  \sigma_{{\rm ou}}^2 \Big( \int_0^{T-t} ds  {\cal K}^\eps(s) {\rm Var} \big( \EE \big[ G'(Z_{s}^\eps) 
|{\cal F}_0 \big]\big)^{1/2} \Big)^2 .
\end{eqnarray*}
The conditional distribution of $Z_t^\eps$ given ${\cal F}_0$ is Gaussian with mean
$$
\EE \big[  Z_t^\eps |{\cal F}_0 \big] = \sigma_{{\rm ou}}  \int_{-\infty}^0 {\cal K}^\eps(t-u) dW_u
$$
and variance
$$
 {\rm Var} \big( Z_t^\eps |{\cal F}_0\big) = (\sigma_{0,t}^\eps)^2 =\sigma_{{\rm ou}}^2 \int_0^{t} {\cal K}^\eps(u)^2 du  .
$$
Therefore,
$$
{\rm Var} \big( \EE \big[ G'(Z_{t}^\eps) |{\cal F}_0 \big]\big)
=
{\rm Var} \Big( \int_\RR G'\big(\EE \big[ Z_{t}^\eps |{\cal F}_0 \big]  +\sigma_{0,t}^\eps z \big) p(z) dz \Big)  .
$$
The random variable $\EE \big[ Z_{t}^\eps |{\cal F}_0 \big] $ is Gaussian with mean zero and variance
\begin{equation*}
(\sigma_{t,\infty}^\eps)^2=\sigma_{{\rm ou}}^2
   \int_t^{\infty} {\cal K}^\eps(u)^2 du ,
\end{equation*} 
so that
\begin{eqnarray}
\nonumber
{\rm Var} \big( \EE \big[ G'(Z_{t}^\eps) |{\cal F}_0 \big]\big)
&=&
 \frac{1}{2} \int_\RR \int_\RR dz dz' p(z) p(z') \int_\RR \int_\RR du du' p(u) p(u') \\
\nonumber &&\times 
\Big[   G'\big(\sigma_{t,\infty}^\eps u +\sigma_{0,t}^\eps z \big) - 
G'\big(\sigma_{t,\infty}^\eps u' +\sigma_{0,t}^\eps z \big)\Big] \\
\nonumber&& \times
\Big[   G'\big(\sigma_{t,\infty}^\eps u +\sigma_{0,t}^\eps z' \big) - 
G'\big(\sigma_{t,\infty}^\eps u' +\sigma_{0,t}^\eps z' \big)\Big] \\
\nonumber&\leq & \|G''\|_\infty^2 (\sigma_{t,\infty}^\eps)^2
\frac{1}{2} \int_\RR \int_\RR du du' p(u) p(u')(u-u')^2 \\
& =  & \|G''\|_\infty^2  (\sigma_{t,\infty}^\eps)^2  .
\label{eq:bornvarcond1}
\end{eqnarray}
Therefore,
\begin{eqnarray*}
{\rm Var}(\vartheta^\eps_{t})^{1/2} &\leq & 
\|G''\|_\infty  \sigma_{{\rm ou}}^2   \int_0^{T-t} ds  {\cal K}^\eps(s) \Big( \int_s^{\infty} du {\cal K}^\eps(u)^2 \Big)^{1/2}
 \\
&\leq & \|G''\|_\infty \sigma_{{\rm ou}}^2 \eps^{1/2} 
 \int_0^{(T-t)/\eps} ds  {\cal K}(s) \Big( \int_s^{\infty} du {\cal K}(u)^2 \Big)^{1/2}
  .
\end{eqnarray*}
Because ${\cal K}(s) \leq 1 \wedge K s^{H-\frac{3}{2}}$, this gives
\begin{equation}
{\rm Var}(\vartheta^\eps_{t})^{1/2} \leq 
C
\left\{ 
\begin{array}{ll}
\displaystyle \eps^{1/2} & \mbox{ if } H < 3/4 ,\\
\displaystyle \eps^{1/2} \ln (\eps) & \mbox{ if } H = 3/4 ,\\
\displaystyle \eps^{2-2H} & \mbox{ if } H > 3/4 ,
\end{array}
\right.
\end{equation}
uniformly in $t\in [0,T]$, for some constant $C$.
This completes the proof of the lemma.
\end{proof}

The random term $\phi^\eps_{t}$ defined by (\ref{def:phit}) has the form
\begin{equation}
\phi_{t,T}^\eps = 
\EE \Big[    \int_t^T G (Z_s^\eps)  ds \big| {\cal F}_t \Big] .
\end{equation}
Here, we write explicitly the argument $T$ (maturity) as we compute the correlations
of these random terms for different maturities.

\begin{lemma}
\label{lem:3}
\begin{enumerate}
\item
For any $t \leq T$, $\phi_{t,T}^\eps$ is a zero-mean random variable with standard deviation of order $\eps^{1-H}$,
\begin{equation}\label{eq:stdev}
\eps^{2H-2}
\EE [ (\phi_{t,T}^\eps)^2] \stackrel{\eps \to 0}{\longrightarrow} \sigma_{\phi}^2 (T-t)^{2H},
\end{equation}
where $\sigma_\phi$ is defined by (\ref{def:sigmaphi}).
\item
The covariance function of $\phi_{t,T}^\eps$ has the following limit for any $t \leq T$, $ t'\leq T'$, with $t\leq t'$,
\begin{equation}\label{eq:cov}
\eps^{2H-2}
\EE [  \phi_{t,T}^\eps \phi_{t',T'}^\eps ] \stackrel{\eps \to 0}{\longrightarrow} \sigma_{\phi}^2 (T-t)^{H}(T'-t')^{H} {\cal C}_\phi(t,t';T,T'),
\end{equation}
where the limit correlation is
\begin{equation}
\label{eq:sCdef}
 {\cal C}_\phi(t,t';T,T') = \frac{\int_{0}^\infty du \big[ (u+r )^{H-\frac{1}{2}}- u ^{H-\frac{1}{2}} \big] \big[ (u+s )^{H-\frac{1}{2}}-(u+q )^{H-\frac{1}{2}} \big]}{\int_0^\infty du  \big[ (1+u)^{H-\frac{1}{2}}-u^{H-\frac{1}{2}} \big]^2} ,
\end{equation}
with
$$
q  = \frac{t'-t}{\sqrt{ (T-t)(T'-t')}} ,\quad \quad r  = \frac{\sqrt{T-t}}{\sqrt{T'-t'}} ,\quad \quad s  = \frac{T'-t}{\sqrt{ (T-t)(T'-t')}} .
$$
\item
As $\eps \to 0$, the random process $\eps^{H-1} \phi_{t,T}^\eps$, $t \leq T$,
converges in distribution (in the sense of finite-dimensional distributions) 
to a Gaussian random process $\phi_{t,T}$, $t \leq T$,
with mean zero and covariance 
$ \eps^{2(H-1)} \EE[\phi_{t,T} \phi_{t',T'}]= \sigma_{\phi}^2 (T-t)^{H} (T'-t')^{H}  {\cal C}_\phi(t,t';T,T') $ for any 
$t\in [0,T]$, $t ' \in [0,T']$, with  $t\leq t'$.
\item  The fourth-order moments of $\eps^{H-1} \phi_{t,T}^\eps$ are uniformly bounded: there exists a constant $K_T$ independent of $\eps$ 
such that
\begin{equation}
\sup_{t \in [0,T]} \EE [ (\phi_{t,T}^\eps)^4]^{1/4} \leq K_T \eps^{1-H}.  
\end{equation}
\end{enumerate}
\end{lemma}
\noindent
Note that the mean square increment of the limit process $\phi_{t,T}$ satisfies, for $t,t+h\in [0,T]$,
\begin{eqnarray}
\nonumber
\EE \big[ (\phi_{t,T}-\phi_{t+h,T})^2\big] &=&
\frac{
1
}{\Gamma(H+\frac{1}{2})^2} 
\int_0^\infty du
 \big[ (T-t-h+u)^{H-\frac{1}{2}}-u^{H-\frac{1}{2}} \big]^2 
 \\
 \nonumber
 &&
- 
 \big[ (T-t+u)^{H-\frac{1}{2}}-(u+h)^{H-\frac{1}{2}} \big]^2
 +
  \big[ (u+h)^{H-\frac{1}{2}}-u^{H-\frac{1}{2}} \big]^2
\\
&=& \frac{
(T-t)^{2H-1}}{\Gamma(H+\frac{1}{2})^2}  h  +o(h), \quad h \to 0 .
\end{eqnarray}
This shows that the limit Gaussian process $\phi_{t,T}$
 has the same local regularity (as a function of $t$) as a standard Brownian motion.
 We also have, for any $t <T\leq T+h$,
 \begin{eqnarray}
\EE \big[ (\phi_{t,T+h}-\phi_{t,T})^2\big] &=&
\frac{
(T-t)^{2H-2}}{(2-2H)\Gamma(H-\frac{1}{2})^2} 
h^2  +o(h^2), \quad h \to 0 .
\end{eqnarray}
This shows that the  limit Gaussian process $\phi_{t,T}$
is smooth (mean square differentiable) as a function of the maturity $T$.

\begin{proof}
Let us fix $T_0>0$.
For $t\in [0,T]$, $t ' \in [0,T']$, with $T,T'\leq T_0$, and $t\leq t'$,
the covariance of $\phi_{t,T}^\eps$ is
\begin{eqnarray*}
{\rm Cov}(\phi_{t,T}^\eps,\phi_{t',T'}^\eps) &=&  
\EE\Big[  
\EE \Big[    \int_t^T G (Z_s^\eps)  ds \big| {\cal F}_t \Big] 
\EE \Big[    \int_{t'}^{T'} G (Z_s^\eps)  ds \big| {\cal F}_{t'} \Big] 
\Big]\\
&=&
\EE\Big[  
\EE \Big[    \int_t^T G (Z_s^\eps)  ds \big| {\cal F}_t \Big] 
\EE \Big[    \int_{t'}^{T'} G (Z_s^\eps)  ds \big| {\cal F}_{t} \Big] 
\Big]\\
&=&
\int_0^{T-t} ds \int_{t'-t}^{T'-t} ds' 
{\rm Cov}\big( \EE \big[ G(Z_s^\eps)|{\cal F}_0\big] ,\EE \big[ G(Z_{s'}^\eps)|{\cal F}_0\big] \big) .
\end{eqnarray*}
Then, proceeding as in the proof of the previous lemma, we obtain
\begin{eqnarray*}
{\rm Var}(\phi_{t,T}^\eps) &\leq &  \Big( \int_0^{T-t} ds 
{\rm Var}\big( \EE \big[ G(Z_s^\eps)|{\cal F}_0\big]\big)^{1/2}
\Big)^2\leq  \|G'\|_\infty^2\Big(  \int_0^{T-t} ds  \sigma_{s,\infty}^\eps  \Big)^2 .
\end{eqnarray*}
Because ${\cal K}(s) \leq 1 \wedge K s^{H-\frac{3}{2}}$, this gives
$$
{\rm Var}(\phi_{t,T}^\eps) \leq C_{T_0} \eps^{2-2H} ,
$$
uniformly in $t  \leq T \leq T_0$, for some constant $C_{T_0}$.
More precisely, for $t\in [0,T]$, $t ' \in [0,T']$, with $T,T'\leq T_0$, and $t\leq t'$,
we have
\begin{eqnarray*}
&& {\rm Cov}(\phi_{t,T}^\eps,\phi_{t',T'}^\eps) 
= \int_0^{T-t} ds \int_{t'-t}^{T'-t} ds' \int_\RR \int_\RR dz dz'p(z) p(z') \\
&& \quad \times
\EE\Big[ 
  G\Big(\sigma_{{\rm ou}} \int_{-\infty}^0 {\cal K}^\eps(s-u) dW_u +\sigma^\eps_{0,s} z \Big)  
 G\Big(\sigma_{{\rm ou}} \int_{-\infty}^0 {\cal K}^\eps(s'-u') dW_{u'} +\sigma^\eps_{0,s'} z' \Big)  
\Big]  .
\end{eqnarray*}
Using the fact that $\left<G\right>=0$,
we can write
\begin{eqnarray*}
 {\rm Cov}(\phi_{t,T}^\eps,\phi_{t',T'}^\eps) 
&=&  \int_0^{T-t} ds \int_{t'-t}^{T'-t} ds' \int_\RR \int_\RR dz dz'p(z) p(z')  \\
&&\times 
\EE\Big[ 
\Big(  G\Big(\sigma_{{\rm ou}} \int_{-\infty}^0 {\cal K}^\eps(s-u) dW_u +\sigma^\eps_{0,s} z \Big) - G( \sigma_{{\rm ou}} z) \Big)\\
&&\quad \times
\Big(  G\Big(\sigma_{{\rm ou}} \int_{-\infty}^0 {\cal K}^\eps(s'-u') dW_{u'} +\sigma^\eps_{0,s'} z' \Big)  - G( \sigma_{{\rm ou}} z') \Big)
\Big]  .
\end{eqnarray*}
Therefore,
\begin{eqnarray*}
 {\rm Cov}(\phi_{t,T}^\eps,\phi_{t',T'}^\eps) 
&=&  \int_0^{T-t} ds \int_{t'-t}^{T'-t} ds' \int_\RR \int_\RR dz dz'p(z) p(z') G'(\sigma_{{\rm ou}}z) G'(\sigma_{{\rm ou}}z') \\
&& \times
\EE\Big[ 
\Big(\sigma_{{\rm ou}} \int_{-\infty}^0 {\cal K}^\eps(s-u) dW_u + (\sigma^\eps_{0,s}- \sigma_{{\rm ou}})z \Big)  \\
&&\quad \times
\Big(\sigma_{{\rm ou}} \int_{-\infty}^0 {\cal K}^\eps(s'-u') dW_{u'} +(\sigma^\eps_{0,s'} - \sigma_{{\rm ou}})z' \Big) 
\Big] + V_3^\eps,
\end{eqnarray*}
up to a term $V_3^\eps$,
which is of order $\eps^{3-3H}$,
\begin{eqnarray*}
V_3^\eps 
&\leq & 2 \|G'\|_\infty  \|G''\|_\infty   \int_0^{T-t} ds \int_0^{T'-t} ds' \int_\RR \int_\RR dz dz'p(z) p(z')\\
 && \times
\EE\Big[ 
\Big(\sigma_{{\rm ou}} \int_{-\infty}^0 {\cal K}^\eps(s-u) dW_u + (\sigma^\eps_{0,s}- \sigma_{{\rm ou}})z \Big)^2  \\
&&\quad \times
\Big|\sigma_{{\rm ou}} \int_{-\infty}^0 {\cal K}^\eps(s'-u') dW_{u'} +(\sigma^\eps_{0,s'} - \sigma_{{\rm ou}})z' \Big|
\Big] \\
&\leq & C \|G'\|_\infty  \|G''\|_\infty   \int_0^{T_0-t} ds \int_0^{T_0-t} ds' \int_\RR \int_\RR dz dz'p(z) p(z')\\
 && \times
\Big(\sigma_{{\rm ou}}^2 \int_{-\infty}^0 {\cal K}^\eps(s-u)^2 du + (\sigma^\eps_{0,s}- \sigma_{{\rm ou}})^2z^2 \Big)  \\
&&\quad \times
\Big(\sigma_{{\rm ou}}^2 \int_{-\infty}^0 {\cal K}^\eps(s'-u')^2 du' +(\sigma^\eps_{0,s'} - \sigma_{{\rm ou}})^2 z'^2 \Big)^{1/2} \\
&\leq & C' \|G'\|_\infty  \|G''\|_\infty  \Big[ \int_0^{T_0-t} ds  \int_\RR dz  p(z)  \Big(
(\sigma^\eps_{s,\infty})^2+ (\sigma^\eps_{0,s}- \sigma_{{\rm ou}})^2z^2 \Big) \Big]^{3/2}\\
&\leq & C' \|G'\|_\infty  \|G''\|_\infty  \Big[ \int_0^{T_0-t} ds  
(\sigma^\eps_{s,\infty})^2+ (\sigma^\eps_{0,s}- \sigma_{{\rm ou}})^2   \Big]^{3/2} .
\end{eqnarray*}
Using $(\sigma^\eps_{s,\infty})^2+ (\sigma^\eps_{0,s})^2= \sigma_{{\rm ou}}^2$ and
 \begin{eqnarray}
\nonumber
| \sigma_{{\rm ou}}-\sigma^\eps_{0,s} |&=&
\sigma_{{\rm ou}}  \Big( 1 -\Big(\int_0^{s/\eps} {\cal K}(u)^2 du \Big)^{1/2}\Big)
= \sigma_{{\rm ou}}  \Big( 1 -\Big(1- \int_{s/\eps}^\infty {\cal K}(u)^2 du \Big)^{1/2}\Big)\\
&\leq & \sigma_{{\rm ou}}  \int_{s/\eps}^\infty {\cal K}(u)^2 du \leq  \sigma_{{\rm ou}}  \Big( 1 \wedge K \big(\frac{s}{\eps}\big)^{2H-2} \Big),
\label{eq:estimesigmaeps0s} 
 \end{eqnarray} 
where the fist inequality follows from
$
\sqrt{ 1-x  }  >  1-x 
$ for $0\leq x \leq 1$,
 we get
 \begin{eqnarray*}
V_3^\eps  &\leq & C' \|G'\|_\infty  \|G''\|_\infty  \Big[ \int_0^{T_0-t} ds  2 \sigma_{{\rm ou}}  (\sigma_{{\rm ou}}-\sigma^\eps_{0,s})  \Big]^{3/2} \\
&\leq & C''  \|G'\|_\infty  \|G''\|_\infty \eps^{3-3H} .
\end{eqnarray*}
This gives
\begin{eqnarray*}
 {\rm Cov}(\phi_{t,T}^\eps,\phi_{t',T'}^\eps) 
&=&  \int_0^{T-t} ds \int_{t'-t}^{T'-t} ds' \int_\RR \int_\RR dz dz'p(z) p(z') G'(\sigma_{{\rm ou}}z) G'(\sigma_{{\rm ou}}z') \\
&&\times
\Big( \sigma_{{\rm ou}}^2 \int_0^{\infty} {\cal K}^\eps(s+u) {\cal K}^\eps(s'+u)  du + (\sigma^\eps_{0,s}- \sigma_{{\rm ou}})
(\sigma^\eps_{0,s'} - \sigma_{{\rm ou}}) z z' \Big)  + V_3^\eps\\
&=& V_1^\eps \left< G'\right>^2 + V_2^\eps \sigma_{{\rm ou}}^2 \left< G''\right>^2 + V_3^\eps,
\end{eqnarray*}
with
\begin{eqnarray*}
V_1^\eps &=& \sigma_{{\rm ou}}^2  \int_0^\infty du \Big( \int_0^{T-t} ds {\cal K}^\eps(s+u) \Big)
 \Big( \int_{t'-t}^{T'-t} ds' {\cal K}^\eps(s'+u)\Big) ,\\
V_2^\eps &=&  \Big( \int_0^{T-t} ds  (\sigma^\eps_{0,s}- \sigma_{{\rm ou}})\Big) \Big( \int_{t'-t}^{T'-t} ds'  (\sigma^\eps_{0,s'}- \sigma_{{\rm ou}})\Big) .
\end{eqnarray*}
Using again (\ref{eq:estimesigmaeps0s}), we find that
$$
V_2^\eps \leq C \eps^{4-4H} ,
$$
while
\begin{eqnarray*}
V_1^\eps &=&  \frac{
1
}{\Gamma(H+\frac{1}{2})^2}
\int_0^\infty \big( (T-t+u)^{H-\frac{1}{2}}-u^{H-\frac{1}{2}}\big)\\
&& \hspace*{1.in} \times \big( (T'-t+u)^{H-\frac{1}{2}}-(u+t'-t)^{H-\frac{1}{2}}\big)
 du  \, \eps^{2-2H}\\
 &&+o(\eps^{2-2H}) .
\end{eqnarray*}
Applying the change of variable
$$
u \to (T-t)^{\frac{1}{2}}(T'-t')^{\frac{1}{2}} u  
$$
 gives the first and second items of the lemma with
$$
\sigma_{\phi}^2 = \frac{
\left<G'\right>^2}{\Gamma(H+\frac{1}{2})^2}
\int_0^\infty \big( (1+u)^{H-\frac{1}{2}}-u^{H-\frac{1}{2}}\big)^2 du,
$$
which is equivalent to (\ref{def:sigmaphi}).

In order to prove the third item, we introduce
\begin{equation}
\check{\phi}^\eps_{t,T} =
 \EE \Big[  \int_t^T Z_s^\eps ds \big| {\cal F}_t \Big] ,
\end{equation}
which is a Gaussian random process with mean zero and covariance, 
for $t\in [0,T]$, $t ' \in [0,T']$, with  $t\leq t'$,
\begin{eqnarray*}
{\rm Cov}\big( \check{\phi}^\eps_{t,T},\check{\phi}^\eps_{t',T'} \big) &=& \int_t^Tds \int_{t'}^{T'} ds'
\EE\Big[ \EE[ Z_s^\eps |{\cal F}_t] \EE[ Z_s^\eps |{\cal F}_{t'}] \Big]\\
&=& \int_t^Tds \int_{t'}^{T'} ds'
\EE \Big[ \EE[ Z_s^\eps |{\cal F}_t]  \EE[ Z_s^\eps |{\cal F}_t] \Big]\\
&=& \sigma_{{\rm ou}}^2
\int_0^{T-t} ds \int_{t'-t}^{T'-t} ds'
\EE\Big[ 
\Big(
 \int_{-\infty}^0 {\cal K}^\eps(s-u) dW_u \Big)\Big(
 \int_{-\infty}^0 {\cal K}^\eps(s'-u) dW_u \Big)
 \Big] \\
&=& \sigma_{{\rm ou}}^2
 \int_0^{\infty} du 
 \Big( \int_0^{T-t} ds  {\cal K}^\eps(s+u)  \Big)  \Big( \int_{t'-t}^{T'-t} ds'  {\cal K}^\eps(s'+u)  \Big) .
\end{eqnarray*}
Therefore, for $t_j \in [0,T_j]$, with  $t_1 \leq \cdots \leq t_n$,
the random vector $(\eps^{H-1}\left<G'\right>\check{\phi}^\eps_{t_1,T_1},\ldots,\eps^{H-1}\left<G'\right>\check{\phi}^\eps_{t_n,T_n})$
converges to a Gaussian random vector with mean $0$ and covariance matrix 
$(\sigma_\phi^2 (T_j-t_j)^H (T_l-t_l)^H {\cal C}_\phi(t_j,t_l;T_j,T_l))_{j,l=1}^n$.
In other words, the random process $(\eps^{H-1}\left<G'\right>\check{\phi}^\eps_{t,T})_{0\leq t\leq T <\infty}$ converges in the sense of
finite-dimensional distributions to a Gaussian process $(\phi_{t,T})_{0\leq t\leq T <\infty}$
 with mean $0$ and covariance function $\EE[\phi_{t,T} \phi_{t',T'}]= \sigma_{\phi}^2 (T-t)^{H} (T'-t')^{H}  {\cal C}_\phi(t,t';T,T') $,
 for $t\in [0,T]$, $t ' \in [0,T']$, with  $t\leq t'$.

Moreover, we have
$$
{\rm Var}\big( \check{\phi}^\eps_{t,T} \big)  = \frac{
1
}{\Gamma(H+\frac{1}{2})^2}
\int_0^\infty \big( (1+u)^{H-\frac{1}{2}}-u^{H-\frac{1}{2}}\big)^2 du \, (T-t)^{2H} \eps^{2-2H}  +o(\eps^{2-2H}).
$$
Similarly,
$$
\EE\big[ \check{\phi}^\eps_{t,T} {\phi}^\eps_{t,T} \big] 
=\frac{
\left<G'\right>}{\Gamma(H+\frac{1}{2})^2}
\int_0^\infty \big( (1+u)^{H-\frac{1}{2}}-u^{H-\frac{1}{2}}\big)^2 du \, (T-t)^{2H} \eps^{2-2H}  +o(\eps^{2-2H}).
$$
As a result,
$$
\eps^{2H-2} \EE\big[ ({\phi}^\eps_{t,T} - \left<G'\right> \check{\phi}^\eps_{t,T}  )^2\big] \stackrel{\eps \to 0}{\longrightarrow} 0,
$$
and the random process $(\eps^{H-1}\left<G'\right>  \check{\phi}^\eps_{t,T} )_{0\leq t\leq T <\infty}$
converges in the sense of
finite-dimensional distributions to a Gaussian process $(\phi_{t,T})_{0\leq t\leq T <\infty}$
 with mean $0$ and covariance function $\EE[\phi_{t,T} \phi_{t',T'}]= \sigma_{\phi}^2 (T-t)^{H} (T'-t')^{H}  {\cal C}_\phi(t,t';T,T') $
for $t\in [0,T]$, $t ' \in [0,T']$, with  $t\leq t'$.
This gives the third item of the lemma.

To prove the fourth item of the lemma, we note that
$$
\phi^\eps_{t,T} = \frac{1}{2} \EE \big[ I^\eps_T | {\cal F}_t \big] - \frac{1}{2} I^\eps_t,
$$
where $I^\eps_t$ is defined by (\ref{def:Iepst}).
Therefore,
$$
\sup_{t \in [0,T]} \EE \big[ (\phi^\eps_{t,T} )^4 \big] \leq \sup_{t \in [0,T]}  \EE\big[ (I^\eps_t )^4 \big], 
$$
and the result follows from Lemma \ref{lem:a}, Eq.~(\ref{eq:boundIepst}).
\end{proof}

\begin{lemma}
\label{lem:4}
Let us define, for any $t \in [0,T]$,
\begin{equation}
\gamma^\eps_t =  \frac{1}{2} \int_0^t \big( (\sigma_s^\eps)^2 -\overline{\sigma}^2\big) \phi^\eps_s ds,
\end{equation}
as in (\ref{def:gammaeps}). We have
\begin{equation}
\label{eq:boundgammaepstilde}
\limsup_{\eps\to 0}
\eps^{H-1} 
\sup_{t\in [0,T]} \EE \big[ (\gamma^\eps_t)^2\big]^{1/2} = 0.
\end{equation}
\end{lemma}

\begin{proof}
Let us define, for any $t \in [0,T]$,
\begin{equation}
\Gamma^\eps_t = \int_t^T \big( (\sigma_s^\eps)^2 -\overline{\sigma}^2\big) \phi^\eps_s ds.
\end{equation}
By the definition (\ref{def:Kt}) of $\phi^\eps_s$, we have
$$
  \Gamma^\eps_t  = 2
\int_t^T ds \int_s^T du  \EE\big[ G(Z^\eps_s)G(Z^\eps_u) |{\cal F}_s\big]  .
$$
Therefore,
\begin{eqnarray*}
\EE 
\big[ (\Gamma^\eps_t)^2\big] &=& 2
\int_t^T ds \int_s^T du \int_s^T ds' \int_{s'}^T du'
\EE\Big[    \EE\big[ G(Z^\eps_s)G(Z^\eps_u) |{\cal F}_s\big]
 \EE\big[ G(Z^\eps_{s'})G(Z^\eps_{u'}) |{\cal F}_{s'}\big] \Big]\\
 &=& 2
\int_t^T ds \int_s^T du \int_s^T ds' \int_{s'}^T du'
\EE\Big[   G(Z^\eps_s)G(Z^\eps_u)  
 \EE\big[ G(Z^\eps_{s'})G(Z^\eps_{u'}) |{\cal F}_{s}\big] \Big] \\
 &=& 
\int_t^T ds \int_s^T du 
\EE\Big[    G(Z^\eps_s)G(Z^\eps_u) 
 \EE\big[ \big(\int_s^T G(Z^\eps_{s'}) ds'\big)^2 \big|{\cal F}_{s}\big] \Big] \\
 &=& 
\int_t^T ds 
\EE\Big[    G(Z^\eps_s)\EE\big[ \int_s^T G(Z^\eps_u) du \big| {\cal F}_s\big]
 \EE\big[ \big(\int_s^T G(Z^\eps_{s'}) ds'\big)^2 \big|{\cal F}_{s}\big] \Big] \\ 
 &\leq &
 \|G\|_\infty \int_t^T ds  \EE\Big[\Big|
 \EE\big[ \big(\int_s^T G(Z^\eps_{s'}) ds'\big)^2 \big|{\cal F}_{s}\big] \Big|^{3/2}\Big] \\ 
 &\leq &
 \|G\|_\infty \int_t^T ds  \EE\Big[\Big| \int_s^T G(Z^\eps_{s'}) ds' \Big|^{3}\Big] \\ 
 &\leq &
 \|G\|_\infty \int_t^T ds  \EE\Big[\Big( \int_s^T G(Z^\eps_{s'}) ds' \Big)^{4}\Big]^{3/4}   ,
 \end{eqnarray*} 
 where  in the first inequality  we use that
\begin{eqnarray*}
\Big| \EE\big[ \int_s^T G(Z^\eps_u) du \big| {\cal F}_s\big] \Big| 
\leq 
\Big| \EE\big[ \big( \int_s^T G(Z^\eps_u) du  \big)^2 \big| {\cal F}_s  \big] \Big| ^{1/2}   ,
\end{eqnarray*} 
which follows from the conditional version of Jensen's inequality. 
It follows   by Lemma \ref{lem:a}
that   $\EE \big[ (\Gamma^\eps_t)^2\big]$    is smaller than 
$K' \eps^{3-3H}$ for some constant $K'$.
This proves
\begin{equation}
\label{eq:boundGammaepstilde}
\limsup_{\eps\to 0}
\eps^{H-1} 
\sup_{t\in [0,T]} \EE \big[ (\Gamma^\eps_t)^2\big]^{1/2} = 0.
\end{equation}
Note that $\gamma^\eps_t$ defined by (\ref{def:gammaeps}) is related to $\Gamma^\eps_t$ through
$$
\gamma^\eps_t = 2\left(\Gamma^\eps_0-\Gamma^\eps_t \right) .
$$
Therefore, Eq.~(\ref{eq:boundGammaepstilde}) also implies
$$
\limsup_{\eps\to 0}
\eps^{H-1} 
\sup_{t\in [0,T]} \EE \big[ (\gamma^\eps_t)^2\big]^{1/2} = 0,
$$
which is the desired result.
\end{proof}

\begin{lemma}
\label{lem:5}
Let us define, for any $t \in [0,T]$,
\begin{equation}
{\eta}^\eps_t = \eps^{1-H} \int_0^t \big(  \sigma_s^\eps  -\widetilde{\sigma}\big) ds,
\end{equation}
as in (\ref{def:etaeps}).
We have
\begin{equation}
\label{eq:boundetaepstilde}
\limsup_{\eps\to 0} \eps^{H-1}
\sup_{t\in [0,T]}  \EE \big[ ( {\eta}^\eps_t )^2\big]^{1/2} = 0.
\end{equation}
\end{lemma}
\begin{proof}
By Lemma \ref{prop:process}, we obtain
\begin{eqnarray*}
\EE \big[ ( {\eta}^\eps_t )^2\big] &=& \eps^{2-2H}
\EE \Big[ \Big( \int_{0}^{t} \big(  \sigma_s^\eps  -\widetilde{\sigma}\big) ds\Big)^2\Big] \\
&=&\eps^{2-2H}
\int_0^t \int_0^t {\rm Cov} \big( F (Z^\eps_s), F(Z^\eps_{s'})\big) ds ds'\\
&=&  \eps^{2-2H}
\big( \left< F^2 \right>-\left< F\right>^2\big) 
\int_0^t \int_0^t {\cal C}_\sigma \big(\frac{s-s'}{\eps}\big)  ds ds'\\
&\leq& K \eps^{2-2H} \int_0^T \int_0^T \big(\frac{|s-s'|}{\eps}\big)^{2H-2}  ds ds' \\
&\leq & K' \eps^{ 4-4H} ,
\end{eqnarray*}
for some constants $K,K'$, because $s^{2H-2}$ is integrable over $(0,T)$,
which gives the desired result.
\end{proof}

\begin{lemma}
\label{lem:6}
Let us define, for any $t \in [0,T]$,
\begin{equation}
{\kappa}^\eps_t = \frac{\eps^{1-H}}{2} \int_0^t \big( (\sigma_s^\eps)^2  -\overline{\sigma}^2\big)  ds,
\end{equation}
as in (\ref{def:kappaeps}).
We have
\begin{equation}
\limsup_{\eps\to 0} \eps^{H-1}
\sup_{t\in [0,T]}  \EE \big[ ( {\kappa}^\eps_t )^2\big]^{1/2} = 0.
\end{equation}
\end{lemma}
\begin{proof}
The proof is similar to the one in Lemma \ref{lem:5}.
\end{proof}

\end{appendices}

\end{document}